\documentclass[12pt, reqno]{amsart}
\usepackage[T2A,T1]{fontenc}
\usepackage[english]{babel}
\usepackage[utf8]{inputenc}
\usepackage{csquotes,fullpage}
\usepackage{amsmath,amssymb,amsthm,amsfonts}
\usepackage{verbatim}
\usepackage[mathscr]{euscript}

\usepackage{floatrow,graphicx,wrapfig}
\theoremstyle{plain}
\newtheorem{theorem}{Theorem}[section]
\newtheorem*{theorem*}{Theorem}
\newtheorem{lemma}[theorem]{Lemma}

\newtheorem{proposition}[theorem]{Proposition}
\newtheorem{definition}{Definition}

\theoremstyle{remark}
\newtheorem*{remark}{Remark}

\newcommand{\RomanNumeralCaps}[1]{\MakeUppercase{\romannumeral #1}}

\DeclareMathOperator{\sgn}{sgn}

\newcommand{\wt}[1]{\widetilde{#1}}
\newcommand{\wtc}[1]{\widetilde{\mathcal{#1}}}
\def\e{\varepsilon}
\def\hom{{\rm hom}}

\DeclareMathOperator*{\dom}{dom}
\DeclareMathOperator{\im}{Im}

\usepackage[pdftex]{hyperref}

\begin{document}

\title{Phase transition in a periodic tubular structure}
\author{Alexander V. Kiselev}
\address{Department of Mathematical Sciences, University of Bath, Claverton Down, Bath, BA2 7AY, United Kingdom}
\email{alexander.v.kiselev@gmail.com, ak2084@bath.ac.uk}

\author{Kirill Ryadovkin}
\address{St. Petersburg Department of Steklov Mathematical Institute
of Russian Academy of Sciences, 27 Fontanka, 191023 St. Petersburg, Russia \emph{and}
St. Petersburg University, 7-9 Universitetskaya Embankment, 199034 St Petersburg, Russia }
%\address{St. Petersburg Department of Steklov Mathematical Institute of RAS, 27 Fontanka Quay, St.Petersburg, 191023 Russia {\sc and} Faculty of Mathematics and Computer Science, St.Petersburg State University, 29 14th line of V.O., St.Petersburg, 199178 Russia  }
%\address{St. Petersburg, Russia}
\email{kryadovkin@gmail.com}

\begin{abstract}
  We consider an $\e$-periodic ($\e\to 0$) tubular structure, modelled as a magnetic Laplacian on a metric graph, which is periodic along a single axis. We show that the corresponding Hamiltonian admits norm-resolvent convergence to an ODE on $\mathbb{R}$ which is
  fourth order at  a discrete set of values of the magnetic potential (\emph{critical points}) and second-order generically. In a vicinity of critical points we establish a mixed-order asymptotics. The rate of convergence is also estimated. This represents a physically viable model of a phase transition as the strength of the (constant) magnetic field increases.

  %fourth-order in vicinity of a discrete set of values of the (constant) magnetic potential and second-order generically. The rate of convergence is also estimated. This represents a physically viable model of a phase transition as the strength of the (constant) magnetic field increases.
\end{abstract}

\maketitle

\section{Introduction}

The research aimed at modelling and engineering metamaterials has been recently brought to the forefront of materials science (see, {\it e.g.}, \cite{Phys_book} and references therein). It is widely acknowledged that these novel materials acquire non-classical properties as a result of a careful design of the microstructure,
%of a composite medium,
which can be assumed periodic with a small enough period cell. The mathematical machinery involved in their modelling must therefore include as its backbone the theory of homogenisation (see
{\it e.g.} \cite{Lions, Bakhvalov_Panasenko, Jikov_book}), which aims at characterising limiting, or ``effective'', properties of small-period composites. A typical problem here is to study the asymptotic behaviour of solutions to equations of the type
\begin{equation}
-{\rm div}\bigl(A^\varepsilon(\cdot/\varepsilon)\nabla u_\e\bigr)-\omega^2u_\e=f,\ \ \ \ f\in L^2({\mathbb R}^d),\quad d\ge1,\qquad \omega^2\notin{\mathbb R}_+,
\label{eq:generic_hom}
\end{equation}
where for all $\varepsilon>0$ the matrix $A^\varepsilon$ is $Q$-periodic, $Q:=[0,1)^d,$ non-negative symmetric, and may additionally be required to satisfy the condition of uniform ellipticity.

The results of {\it e.g.}, \cite{Tip_1998,Figotin_Schenker_2005,Tip_2006,Figotin_Schenker_2007b} relate this area of research to the analysis of the so-called time-dispersive media; this relationship was recently quantified in \cite{KCher,CEK,CherErKis,CEKN}.

A prominent example of a problem of the type \eqref{eq:generic_hom} is provided by the homogenisation theory in the uniformly strongly elliptic setting ({\it i.e.}, both $A^\e$ and $(A^\e)^{-1}$ are bounded uniformly in $\varepsilon$). Here one proves (see \cite{Zhikov_1989, BSu03} and references therein) the existence of a constant matrix $A^\hom$ such that solutions $u_\e$ to \eqref{eq:generic_hom} converge to $u_\hom$ satisfying
\begin{equation}
-{\rm div}\bigl(A^\hom \nabla u_\hom\bigr)-\omega^2u_\hom=f,
\label{eq:generic_hom1}
\end{equation}
with the result also carrying over to vector models, including the Maxwell system.

It has been well understood that under some non-restrictive additional conditions the effective parameters $A^{\rm hom}$ in  (\ref{eq:generic_hom1}) are given by the leading-order term at the zero energy $\lambda=0$ of the energy-quasimomentum dispersion relation
$\lambda_1=\lambda_1^\varepsilon(
\varkappa
)=
%\varepsilon^{-2}A_{\rm hom}\theta\cdot\theta
A^{\rm hom}\varkappa\cdot\varkappa
%+O(\varepsilon^{-3}\theta^3)
+O(\varkappa^3),$ $\varkappa\to0,$ for the first eigenvalue in the problem
\begin{equation}
-\bigl(\nabla+{\rm i}
%\varepsilon^{-1}\theta
\varkappa\bigr)\cdot A^\varepsilon\bigl(\nabla+{\rm i}
%\varepsilon^{-1}\theta
\varkappa
\bigr)u=\lambda u,\ \ \ \ u\in L^2(Q),\ \ \ \varkappa\in\bigl[0,2\pi\bigr)^d,
\label{fibre_equation}
\end{equation}
%in the spac$L^2(Q)$
with respect to the scaled variable $y=x/\varepsilon\in Q,$ so that $A^\varepsilon=A^\varepsilon(y),$ and the gradient $\nabla$ in (\ref{fibre_equation}) is taken with respect to
$y.$   The direct fibre decomposition into problems (\ref{fibre_equation}), followed by a perturbation analysis of its eigenvalue
%$\lambda^\varepsilon(\theta)$
$\lambda_1^\varepsilon(\varkappa)$
in each fibre,
%of the $\varepsilon$-dependent operator
allows one to obtain sharp operator-norm resolvent convergence estimates for the problem (\ref{eq:generic_hom}), see \cite{Zhikov_1989}, \cite{BSu03}.

It is notable however that the fact that $\lambda_1^\varepsilon(\varkappa)$ is quadratic in $\varkappa$ at the lower edge of the spectrum does not necessarily hold in the problems associated with the \emph{magnetic} Schrodinger operator. As observed in \cite{Shterenberg}, for a specially constructed magnetic potential and in the setting of $\mathbb{R}^2$ one can, for a strong enough magnetic field, face a situation where the asymptotic expansion of $\lambda_1^\varepsilon(\varkappa)$ for small $\varkappa$ starts with a fourth-order term. The general argument explained above then seems to suggest that the effective operator of \eqref{eq:generic_hom1} must therefore be replaced by a fourth-order differential operator.

The model of \cite{Shterenberg} can thus be viewed as that exhibiting a phase transition from a differential operator of the second order to the one of the fourth order as the magnetic potential reaches a certain threshold. We mention that the example considered in \cite{Shterenberg} can hardly be seen as physical as the assumption that the magnetic potential has a jump over a curve in the plane seems to equate to the consent for magnetic monopoles or equally extravagant objects.

In the present paper we therefore construct an explicitly solvable model exhibiting precisely the same type of a phase transition. The model is in fact a tubular $\e$-periodic structure akin to a carbon nanotube, periodic along a single axis. We are then able to not only study the phase transition phenomenon itself, but also to ascertain that under some non-restrictive geometric assumptions the magnetic potential required can be obtained by switching on a \emph{constant} magnetic field.

Thin networks are customarily modelled by quantum graphs, which is based on the norm-resolvent convergence results pertaining to Laplacians on thin manifolds, as studied in full details in \cite{Postbook}, see also references therein; cf. \cite{Pavlov_2008,MPP} where a related problem formulated in terms of scattering theory is considered and \cite{antiPost} where an alternative approach to the problem is pursued. In the present paper, we duly rely on this research in selecting a quantum graph equipped with magnetic Laplacian operators on its edges as a realistic mathematical model of a thin network. The matching conditions at the graph vertices are chosen to be Kirchhoff, or standard ones (see \cite{Berkolaiko}). This means that the original thin manifold has non-resonant vertex volumes, i.e., the ratio of the volume of vertex parts of the manifold to that of its edge parts is assumed to vanish as the manifold ``converges'' to the metric graph. The resonant case, where the mentioned ratio is bounded below, will be treated elsewhere.

We further clarify that as the period $\e$ of the periodic tubular structure tends to zero, in order to attain the desired effect one has to allow, perfectly in line with \cite{Shterenberg}, the magnetic potential to scale accordingly, i.e., it has to grow as $\e^{-1}$. Under this assumption, the lower edge of the spectrum of the operator family considered grows  to $+\infty$ as $\e^{-2}$, which is the so-called high-frequency homogenisation regime in terms of \cite{Craster}, see also pioneering works \cite{BirmanHF1,BirmanHF2}. Moreover, the techniques developed in this paper are applicable with no significant modifications in the treatment of the homogenisation problem in a vicinity of an internal edge of the spectrum (using the terminology of \cite{BirmanHF1,BirmanHF2}), termed as the high-frequency homogenisation regime in \cite{Craster}. Therefore the present paper introduces an alternative approach in the mentioned area, permitting to prove norm-resolvent convergence with an explicitly controlled rate as opposed to the approach of \cite{Craster} on the one hand, and suitable for non-factorizable operator families as opposed to that of \cite{BirmanHF1,BirmanHF2} on the other. The argument we develop here can also be extended by constructing an explicit functional model of the operators of the class considered, paving the way to the study of the phase transition in terms of the spectral function. This will be done in a separate publication; here we will be concerned exclusively with norm-resolvent convergence to the effective (homogenised) operator.

The paper is organised as follows. In Section \ref{sec:Preliminaries} we introduce the notation pertaining to metric graphs and define the operator family under consideration. In Section \ref{sec:Fibre_representation} the fibre representation of the operator is constructed based on a convenient form of Gelfand transform. In Section \ref{sec:Boundary_triple} the apparatus of the boundary triples theory is recalled and applied to the operators on graphs. Section \ref{sec:Model_graph} introduces the model graph we analyse in the sequel. The intended simplicity of this model allows for a complete spectral analysis of the associated operator, which we have elected to include as well. Further, Section \ref{sec:Rescaling} rescales the model graph, yielding the $\e$-periodic structure to be homogenised next. Section \ref{sec:Degeneracy} then treats the problem of degenerate band edges, culminating in the complete description of all magnetic potentials leading to degeneracy. Section \ref{sec:Perturbation} provides asymptotics of the $M$-matrix, which is then used in Section \ref{sec:Homogenisation} to prove the main result of the paper. The latter is Theorem \ref{thm:MAIN} and in a simplified form reads:
\begin{theorem*}
There exists a periodic sequence $\mathscr A$ of values of the magnetic potential which are critical in the following sense. Let $A\in \mathscr A$,
 $A'=A+\delta$ and $z'_\e$ be the lower edge of the spectrum of the self-adjoint magnetic Laplacian ${\Delta}_\e(A'/\e)$ in the graph space $L^2(\mathcal{G}^\e_1)$  describing the tubular structure (see details in Section 5).

Let $\mathcal{A}_-$ and $\mathcal{A}_+$ be self-adjoint operators in $L^2(\mathbb R)$ defined by the differential expressions
$$
\kappa_2^2  \delta \frac{ d^2}{dx^2}+\kappa_4^2 \frac{d^4}{dx^4}
\quad \text{and}\quad
\kappa_4^2 \frac{d^4}{dx^4},
$$
respectively, where $\kappa_2$ and $\kappa_4$ are explicitly computed real constants.

Then the following estimates hold uniformly in $z\in K_\sigma$, a compact set in $\mathbb{C},$ $\sigma$-uniformly separated from the real line, with a unitary operator $\Phi:L^2(\mathcal{G}^\e_1)\mapsto L^2(\mathbb{R})$:
  $$
  \|({\Delta}_\e(A'/\e)-(z'_\e+z))^{-1}-\Phi^*(\mathcal A_{-}-2\kappa_0^{-2}z)^{-1}\Phi\|= O((\min\{\varepsilon^{1/2},\varepsilon/|\delta|\})+O(|\delta|),\quad \delta\leqslant0,
  $$
and
  $$
  \|({\Delta}_\e(A'/\e)-(z'_\e+z))^{-1}-\Phi^*(\mathcal A_{+}-2\kappa_0^{-2}z)^{-1}\Phi\|= O(\varepsilon^{1/2}),\quad \delta=O(\e^2)>0.
  $$
Here $\kappa_0$ is an explicitly computed real constant.

\end{theorem*}

\section{Preliminaries}\label{sec:Preliminaries}

\subsection{Periodic graphs}
Let $\mathcal{G}=(\mathcal{V},\mathcal{E})$ be a oriented connected graph (the cases of loops and multiple edges are not excluded). Here
 $\mathcal{V}$ is the set of its vertices, whereas  $\mathcal{E}$ denotes the set of its oriented edges, i.e., the set of ordered pairs of vertices belonging to $\mathcal{V}$,
\[
\mathcal{E}=\{(v,u) \,\vert\, v,u\in\mathcal{V}\}.
\]
We will further assume that the graph  $\mathcal{G}$ is embedded into  $\mathbb{R}^{d_0}$ for some  $d_0\geqslant 1$, i.e., $\mathcal{V}\subset\mathbb{R}^{d_0}$.
Each of the edges $e\in\mathcal{E}$ is represented as an interval  $[0,l_e]$ of a given length $l_e\in(0,+\infty)$. One might consider such graph as a collection $\mathcal{V}$ of points in $\mathbb{R}^{d_0}$, which correspond to the vertices, some of which are connected by smooth curves, which correspond to its edges. Since any graph embedded into  $\mathbb{R}^{d_0}$  is also naturally embedded into  $\mathbb{R}^{d_1}$ for all $d_1>d_0$, we will assume without loss of generality that the smooth curves corresponding to the graph edges have no other common points than the graph vertices.

In the present paper we will only consider \textbf{locally finite} and \textbf{periodic} for some $1\leq d\leq d_0$ graphs in $\mathbb{R}^d$. Namely, a graph  $\mathcal{G}=(\mathcal{V},\mathcal{E})$ is called \textbf{locally finite}, if any bounded set  $K\subset\mathbb{R}^{d_0}$ contains at most a finite number of vertices from $\mathcal{V}$ and the multiplicity (i.e., the total number of incoming and outgoing edges) of all vertices is finite.
Let $g_1,\ldots,g_d$ be a basis in $\mathbb{R}^d$. We will refer to the set
\begin{equation*}
\Gamma=\{g\in\mathbb{R}^d \,\vert\, g=\sum\limits_{j=1}^{d}n_jg_j, \,\,\, n_j\in\mathbb{Z} \,\,\, \forall j=1,\ldots,d \}
\end{equation*}
as the lattice $\Gamma$, generated by the basis $g_1,\ldots,g_d$.
A graph $\mathcal{G}=(\mathcal{V},\mathcal{E})$ is called  \textbf{periodic} in $\mathbb{R}^d$ with periods $g_1,\ldots,g_d$, if for any  $g\in\Gamma$ the graphs $\mathcal{G}$ and $\mathcal{G}+g$ coincide. Here $\mathcal{G}+g$ is understood as the graph $\mathcal{G}$, shifted by the vector $g$
\begin{align*}
\mathcal{G}+g&=\{\mathcal{V}+g,\mathcal{E}+g\},
\\
\mathcal{V}+g&=\{v\,\vert\,v-g\in\mathcal{V}\},
\\
\mathcal{E}+g&=\{(v,u)\,\vert\, (v-g,u-g)\in\mathcal{E}\}.
\end{align*}

Denote by $\mathcal{C}$ the parallelepiped, defined by the vectors $g_1,\ldots,g_d$
\[
\mathcal{C}=\{x\in\mathbb{R}^d \,\vert\, x=\sum\limits_{j=1}^d x_j g_j, \,\,\, 0\leqslant x_j < 1 \,\,\, \forall j=1,\ldots,d \}.
\]
Let $\widetilde{\mathcal{V}}$ be the set of vertices of the graph  $\mathcal{G}$ which belong to this parallelepiped,
\[
\widetilde{\mathcal{V}}=\{ \widetilde{v}\in \mathcal{V}: \widetilde{v}\in\mathcal{C} \}.
\]
Then any vertex  $v\in \mathcal{V}$ admits a unique representation
\[
v=\widetilde{v}+g_v, \qquad \widetilde{v}\in \widetilde{\mathcal{V}}, \quad g_v\in\Gamma.
\]
The sets $\mathcal{C}$, $\widetilde{\mathcal{V}}$ are commonly referred to as the elementary cell and fundamental vertex set, respectively.
We denote by $\widetilde{\mathcal{G}}$ the graph $\mathcal{G}$ factored over the translation group by vectors from $\Gamma$. We will henceforth identify the vertices of the graph $\widetilde{\mathcal{G}}$, being equivalence classes of vertices of  $\mathcal{V}$, with the set  $\widetilde{\mathcal{V}}$ containing precisely one member of each such class.

We further denote by  $\widetilde{\mathcal{E}}$ the set of edges of the graph $\widetilde{\mathcal{G}}$. The edge $\widetilde{e}=(\widetilde{v},\widetilde{u})$ belongs to the set $\widetilde{\mathcal{E}}$ if and only if there exists $g_u\in\Gamma$ such that the graph $\mathcal{G}$ contains the edge $(\widetilde{v},\widetilde{u}+g_u)$. The multiplicity of the edge $\widetilde{e}=(\widetilde{v},\widetilde{u})$ is the sum of multiplicities of the edges $e=(\widetilde{v},\widetilde{u}+g_u)$ over all $g_u\in\Gamma$. We remark that even when the graph $\mathcal{G}$ contains no multiple edges, such edges might appear in its fundamental graph
$\widetilde{\mathcal{G}}$.

In what follows, $\widetilde{e}$ will denote the edge of the graph $\widetilde{\mathcal{G}}$ such that the edge  $e$ of the graph $\mathcal{G}$ belongs to its equivalence class. As in the case of the set of vertices of the fundamental graph, it would be convenient to us to identify its edges with some subset of mutually non-equivalent edges of the original graph. We resort to this abuse of notation in one case only, namely, if the edge $\widetilde{e}=(\widetilde{v},\widetilde{u})$ of the graph $\widetilde{\mathcal{G}}$ corresponds to the edge $e=(\widetilde{v},\widetilde{u}+g_u)$ of the graph $\mathcal{G}$, then the notation $\widetilde{e}+g$ will be understood as representing the edge $(\widetilde{v}+g,\widetilde{u}+g_u+g)$ of the graph $\mathcal{G}$. This corresponds to the identification of the edge $\widetilde{e}$ of the graph $\widetilde{\mathcal{G}}$ with the edge $e$ of $\mathcal{G}$ belonging to its equivalence class and originating in the fundamental vertex set.

We will assume that the length of  $\widetilde{e}$ on the graph $\widetilde{\mathcal{G}}$ is equal to $l_e$, the length of the edge  $e$ on the graph $\mathcal{G}$.
Let $e=(v,u)$. We will denote by $g_e$ the following vector of the lattice $\Gamma$:
\begin{equation*}
\label{eq:dge}
g_e=g_u-g_v.
\end{equation*}
We remark that the vector $g_e$ is invariant for all edges in the equivalence class $\widetilde{e}$.

\medskip

\subsection{The operator $\Delta(A)$}
For a sufficiently smooth function $f$ defined on the graph $\mathcal{G}$ we denote its restriction to the edge  $e$ by $f_e$. Let  $W^k_2(\mathcal{G})$ be the Sobolev space associated with the graph $\mathcal{G}$, i.e., the space of functions which on each edge have square summable distributional derivatives up to the order $k$. The (standard) norm in this space is
\[
\Vert f \Vert^2_{W^k_2(\mathcal{G})}=\sum\limits_{e\in\mathcal{E}}\Vert f_e \Vert^2_{W^k_2(0,l_e)}=\sum\limits_{e\in\mathcal{E}}\sum\limits_{j=0}^{k}\int\limits_{0}^{l_e}\vert f^{(j)}_e(x) \vert^2\,dx.
\]

%Будем считать, что $p\geq1/2$, тогда функции из $W^p_2(\mathcal{G})$ абсолютны непрерывны на каждом из ребер.
We will say that a function $f$ is continuous on the graph $\mathcal{G}$ if it is continuous on each of its edges and at each vertex $v\in\mathcal{V}$ there exists a value $f(v)$ such that for all $u$ and $w$ satisfying $(u,v), (v,w)\in\mathcal{E}$ the following equality holds:
\[
f(v)=f_{(u,v)}(l_e)=f_{(v,w)}(0),
\]
i.e., the value of the function at $v$ coincides with the values of it at the origins of outgoing and at the ends of incoming edges.

Let now $b,A$ be two bounded functions on  $\mathcal{G}$, periodic with respect to the lattice  $\Gamma$, i.e., $b_{e+g}=b_e$, $A_{e+g}=A_e$ for all  $g\in\Gamma$.
Assuming $f\in W_2^2(\mathcal{G})$, each of the functions $f_e$ has an absolutely continuous derivative. Denote by $\partial^A_n f_e(v)$ the co-normal derivative of the function $f$ at the vertex $v\in\mathcal{V}$ along the edge $e\in\mathcal{G}$,
\begin{equation*}
\label{eq:ddn}
\partial^A_n f_e(v)=
\begin{cases}
b^2_e(0)\big(f'(0)+i A_e(0)f(0)\big), &\quad e=(v,u);
\\
-b^2_e(l_e)\big(f'(l_e)+i A_e(l_e)f(l_e)\big), &\quad e=(u,v);
\\
0,\quad &\quad\text{otherwise}.
\end{cases}
\end{equation*}

In order not to overcomplicate the exposition with technical details, we will henceforth assume that $A_e,\ b_e\in C^1(e)$  and that the weights $b_e$ are non-degenerate, i.e., $b_e\ge c_0>0$ with an independent constant $c_0$, for each edge $\e\in\mathcal E$.

By $\Delta(A)$ we denote the operator corresponding to the differential expression
\begin{equation}
\label{eq:dDelta}
(\Delta(A)f)_e(x)=-\Big(\frac{d}{dx}+i A_e(x)\Big)b^2_e(x)\Big(\frac{d}{dx}+i A_e(x)\Big)f_e(x),
\end{equation}
acting on the domain
\begin{align*}
\label{eq:ddomA}
\dom \Delta(A) = \{f\in W_2^2(\mathcal{G}) \,\vert\, \, \text{ $f$ is continuous }
\text{on $\mathcal{G}$ and} \sum\limits_{e\in\mathcal{E}}\partial^{A}_n f_e(v)=0 \text{ for every } v\in\mathcal{V}\}.
\end{align*}
It is easily checked that this operator is self-adjoint in $L^2(\mathcal{G}):=W_2^0(\mathcal{G})$; this is also an immediate corollary of the general theory, see Section \ref{sec:Boundary_triple} below.

\medskip

\section{Fibre representation}\label{sec:Fibre_representation}
Let the basis $\{\widetilde{g}_j\}_{j=1}^d$ be dual to  $\{g_j\}_{j=1}^d$, i.e.,
\begin{equation*}
\langle \widetilde{g}_i , g_j \rangle = 2 \pi \delta_{ij}.
\end{equation*}
This basis gives rise to the lattice $\wt{\Gamma}$, dual to  $\Gamma$:
\begin{equation*}
\wt{\Gamma}=\{g\in\mathbb{R}^d \,\vert\, g=\sum\limits_{j=1}^{d}n_j\wt{g}_j, \,\,\, n_j\in\mathbb{Z} \,\,\, \forall j=1,\ldots,d \}.
\end{equation*}
We pick the first Brillouin zone (see, e.g., \cite{BSu03}) as an elementary cell of the dual lattice:
\begin{equation*}
%\label{eq:dCtil}
\widetilde{{\mathcal C}}=\{t\in\mathbb{R}^d : \vert t \vert \leq \vert t-g \vert, \,\, 0\not=g\in\wt{\Gamma}\}.
\end{equation*}
We remark that as an elementary cell of the dual lattice one could pick any set which is fundamental for the latter (see \cite{BSu03}). In particular, the simplest, albeit not the most natural, choice is
\begin{equation*}
\{t\in\mathbb{R}^d : t=\sum\limits_{j=1}^{d}t_j\widetilde{g}_j,
\,\, 0\leqslant t_j < 1,\,j=1,\ldots,d\}.
\end{equation*}
%\end{remark}

Denote by $\mathcal{L}$ the direct integral of Hilbert spaces   $L_2(\widetilde{\mathcal{G}})$ over the dual cell $\widetilde{\mathcal{C}}$,
\[
\mathcal{L}=\int\limits_{\widetilde{\mathcal{C}}}\oplus L_2(\widetilde{\mathcal{G}})\,dt.
\]
This is a space of functions of two variables, $x\in\widetilde{\mathcal{G}}$ and $t\in\widetilde{\mathcal{C}}$, such that for a.a. $t$ they belong to $L_2(\widetilde{\mathcal{G}})$ as functions of $x$. The norm in this space is defined by
\[
\Vert f \Vert_{\mathcal{L}}^2=\int\limits_{\widetilde{\mathcal{C}}}\Vert f \Vert_{L_2(\widetilde{\mathcal{G}})}^2\,dt.
\]

Let $L_2^0({\mathcal{G}})$ be the set of compactly supported functions from  $L_2({\mathcal{G}})$.
We define the operator $U:L_2^0({\mathcal{G}})\to\mathcal{L}$ as follows:
\begin{equation}
\label{eq:dU}
(Uf)_{\widetilde{e}}(x,t)=\frac{1}{\sqrt{\vert \widetilde{\mathcal{C}}\vert}}\sum\limits_{g\in\Gamma}f_{\wt{e}+g}(x)e^{-i\langle g,t \rangle} e^{-\frac{ix}{l_e}\langle g_e,t\rangle}.
\end{equation}
We recall that this definition does not depend on which particular edge $e$ of the equivalence class $\widetilde{e}$ is picked. Operators of this type are commonly referred to as Gelfand transforms.
 \begin{lemma}
The operator $U$ can be extended by continuity to a unitary operator acting from $L_2({\mathcal{G}})$ to $\mathcal{L}$.
\end{lemma}
\begin{proof}
Pick a function $f\in L_2^0({\mathcal{G}})$, i.e., a square summable on $\mathcal{G}$ function with compact support in $\mathbb{R}^{d_0}$. The norm of $Uf$ admits the representation
\begin{align*}
\Vert Uf \Vert^2_{\mathcal{L}}&=\frac{1}{\vert \widetilde{\mathcal{C}}\vert}\sum\limits_{\widetilde{e}\in\widetilde{\mathcal{E}}}\int\limits_{\widetilde{\mathcal{C}}}\int\limits_{0}^{l_e} \sum\limits_{g_1\in \Gamma}\sum\limits_{g_2\in \Gamma} f_{\wt{e}+g_1}(x) \overline{f_{\wt{e}+g_2}(x)} e^{-i\langle g_1-g_2,t \rangle}\, dx \,dt
\\&=
\sum\limits_{\widetilde{e}\in\widetilde{\mathcal{E}}}\int\limits_{0}^{l_e}\sum\limits_{g\in \Gamma}f_{e+g}(x) \overline{f_{e+g}(x)} \, dx=\sum\limits_{e\in\mathcal{E}}\Vert f_e \Vert^2_{L_2(0,l_e)}=\Vert f \Vert^2_{L_2(\mathcal{G})}.
\end{align*}
Hence, $U$ is extended by continuity to an isometry from  $L_2({\mathcal{G}})$ to $\mathcal{L}$.

The operator adjoint to $U$ can be defined by the equality
\[
(U^*h)_{\wt{e}+g}(x)=\frac{1}{\sqrt{\vert\widetilde{\mathcal{C}}\vert}}\int\limits_{\widetilde{\mathcal{C}}}h_{\widetilde{e}}(x,t)e^{i\langle g,t \rangle}e^{\frac{ix}{l_e}\langle g_e,t\rangle}\,dt
\]
for any $h\in\mathcal{L}$.
Note that for a fixed  $x$ the integral in the latter expression is precisely a coefficient of the Fourier series for $h_{\widetilde{e}}(x,t) e^{\frac{ix}{l_e}\langle g_e,t\rangle}$. The Parceval identity then yields, for a.a. $x\in(0,l_{e})$,
\[
\sum\limits_{g\in\Gamma}\Big\vert\int\limits_{\widetilde{\mathcal{C}}}h_{\widetilde{e}}(x,t)e^{i\langle g,t\rangle}e^{\frac{ix}{l_e}\langle g_e,t\rangle}\,dt\Big\vert^2
=\vert \widetilde{\mathcal{C}}\vert\int\limits_{\widetilde{\mathcal{C}}}\vert h_{\widetilde{e}}(x,t)\vert^2\,dt.
\]
We therefore obtain for the norm of the function $U^*h$:
\begin{align*}
\Vert U^*h \Vert^2_{L_2(\mathcal{G})}&=\sum\limits_{e\in\mathcal{E}}\int\limits_{0}^{l_e}\vert (U^*h)_e(x)\vert^2\, dx
= \frac{1}{\vert \widetilde{\mathcal{C}}\vert}\sum\limits_{\widetilde{e}\in\widetilde{\mathcal{E}}}\sum\limits_{g\in\Gamma}\int\limits_{0}^{l_e} \Big\vert\int\limits_{\widetilde{\mathcal{C}}}h_{\widetilde{e}}(x,t)e^{i\langle g,t\rangle}e^{\frac{ix}{l_{e}}\langle g_e,t\rangle}\,dt\Big\vert^2 \, dx
\\&=
\sum\limits_{\widetilde{e}\in\widetilde{\mathcal{E}}}
\int\limits_{0}^{l_e}\int\limits_{\widetilde{\mathcal{C}}}\vert h_{\widetilde{e}}(x,t)\vert^2\,dt dx=
\int\limits_{\widetilde{\mathcal{C}}}\Vert h(\cdot,t)\Vert^2_{L_2(\widetilde{\mathcal{G}})}\,dt=\Vert h\Vert^2_{\mathcal{L}}.
\end{align*}
Thus, the operator adjoint to $U$ is an isometry and hence the operator $U$, extended to the space $L_2({\mathcal{G}})$ by continuity, is unitary.
\end{proof}

Denote by  $b_{\widetilde{e}},A_{\widetilde{e}}$ functions on the edge  $\wt{e}$ of the graph $\widetilde{\mathcal{G}}$ such that for any graph edge $e\in\mathcal{E}$ corresponding to $\wt{e}$ one has
\[
b_{\widetilde{e}}(x)=b_e(x), \qquad A_{\widetilde{e}}(x)=A_e(x).
\]
Then
\begin{equation}
\label{eq:wtaU}
(U b f)_{\widetilde{e}}(x,t)=b_{\widetilde{e}}(x)(Uf)_{\widetilde{e}}(x,t).
\end{equation}

Further, let $f\in W^2_2(\widetilde{\mathcal{G}})$.
Define $\widetilde{\partial}^A_n f_{\widetilde{e}}(\widetilde{v},t)$ by the following equality:
\begin{equation*}
\label{eq:ddwtn}
\widetilde{\partial}^A_n f_{\widetilde{e}}(\widetilde{v},t)=
\begin{cases}
b^2_{\widetilde{e}}(0)\big(f'_{\widetilde{e}}(0)+i\big(\frac{\langle g_e,t\rangle}{l_{e}}+A_{\widetilde{e}}(0)\big) f_{\widetilde{e}}(0)\big),
&\quad \widetilde{e}=(\widetilde{v},\widetilde{u});
\\
-b^2_{\widetilde{e}}(l_{e})\big(f'_{\widetilde{e}}(l_{e})+i\big(\frac{\langle g_e,t\rangle}{l_{e}} +A_{\widetilde{e}}(l_{e})\big)f_{\widetilde{e}}(l_{e})\big),
&\quad \widetilde{e}=(\widetilde{u},\widetilde{v});
\\
0,\quad &\quad\text{otherwise}.
\end{cases}
\end{equation*}
Let $\wt{\Delta}(A,t)$ be defined as the operator in $\mathcal{L}$ acting as
\begin{equation}
\label{eq:dwtA}
(\wt{\Delta}(A,t)f)_{\widetilde{e}}(x,t)=-\Big(\frac{d}{dx}+\frac{i}{l_{e}}\langle g_e,t\rangle+iA_{\wt{e}}(x)\Big)b^2_{\widetilde{e}}(x)\Big(\frac{d}{dx}+\frac{i}{l_{e}}\langle g_e,t\rangle+iA_{\wt{e}}(x)\Big) f_{\widetilde{e}}(x,t)
\end{equation}
on the domain
\begin{multline}
\label{eq:ddomwtA}
\dom \wt{\Delta}(A,t) = \{f\in W_2^2(\widetilde{\mathcal{G}}) \,\vert\, \, \text{the function $f$ is continuous } \\
\text{on the graph $\wtc{G}$, and } \sum\limits_{\wt{e}\in\wtc{E}}\wt{\partial}^A_n f_{\wt{e}}(\wt{v},t)=0 \text{ for every } \wt{v}\in  \wtc{G} \}.
\end{multline}

We have the following
\begin{theorem}
\label{th:direct}
The operator $\Delta(A)$ defined by \eqref{eq:dDelta} is unitary equivalent to the direct integral of the operators $\wt{\Delta}(A,t)$ defined by \eqref{eq:dwtA} on the domain \eqref{eq:ddomwtA},
\begin{equation}\label{eq:fiber}
U\Delta(A)U^{-1}=\int\limits_{\widetilde{\mathcal{C}}}\oplus \wt{\Delta}(A,t) \,d t.
\end{equation}
\end{theorem}
\begin{proof}
Let $f(\cdot)\in\dom \Delta(A)$. We will demonstrate that $(Uf)(\cdot,t)\in\dom \wt{\Delta}(A,t)$.
The continuity condition for $Uf(\cdot,t)$ at the vertices of the graph $\widetilde{\mathcal{G}}$ follows from the continuity of $f$ at the vertices of the graph $\mathcal{G}$ together with the definition \eqref{eq:dU} of the operator $U$. For the Gelfand transform  \eqref{eq:dU} of $f'$ we have
\begin{align*}
(Uf')_{\widetilde{e}}(x,t)&=\frac{1}{\sqrt{\vert \widetilde{\mathcal{C}}\vert}}\sum\limits_{g\in\Gamma}f'_{\wt{e}+g}(x)e^{-i\langle g,t \rangle} e^{-\frac{ix}{l_e}\langle g_e,t\rangle}
\\&=
\frac{1}{\sqrt{\vert \widetilde{\mathcal{C}}\vert}}\sum\limits_{g\in\Gamma}\Big(\frac{d}{dx}+\frac{i}{l_{e}}\langle g_e,t\rangle\Big)f_{\wt{e}+g}(x)e^{-i\langle g,t \rangle} e^{-\frac{ix}{l_e}\langle g_e,t\rangle}
\\&=
\Big(\frac{d}{dx}+\frac{i}{l_{e}}\langle g_e,t\rangle\Big) (Uf)_{\widetilde{e}}(x,t).
\end{align*}
Therefore, if  $f$ satisfies the second condition on the domain of $\Delta(A)$, $Uf$ satisfies the second condition on the domain of the operator $\wt{\Delta}(A,t)$.

A similar calculation for the second derivative together with \eqref{eq:wtaU} yields
\[
(U\Delta(A)f)(x,t)=\wt{\Delta}(A,t)(Uf)(x,t).
\]
\end{proof}

\section{The boundary triple}\label{sec:Boundary_triple}

We first recall the setup of the boundary triples theory, which we will use next.
The fundamentals of the boundary triples theory
are covered
in~\cite{BHS2020, Smudgen2012}, see also
references therein.
%
%Here we report only the results required.

%

Denote by~$\mathcal A$ a closed and densely defined symmetric
operator on the separable Hilbert space $H$ with the
domain~$\dom\mathcal A$, having equal deficiency
indices~$0 < n_{+}(\mathcal A)= n_{-}(\mathcal A) \leq\infty$.

\begin{definition}[\cite{Kochubej}]\label{def:bonudaryTriple}
A triple $\{\mathcal{K}, \Gamma_0, \Gamma_1 \}$ consisting of an auxiliary
Hilbert space $\mathcal{K}$ and linear mappings $\Gamma_0, \Gamma_1$ defined
everywhere on  $\dom \mathcal A^*$ is called a \emph{boundary triple} for $\mathcal A^*$ if the following
conditions are satisfied:
\begin{enumerate}
\item
The abstract Green's formula is valid
    \begin{equation*}\label{GreenFormula}
        (\mathcal A^*f,g)_H - (f,\mathcal A^*g)_H = (\Gamma_1 f, \Gamma_0 g)_{\mathcal{K}} -
        (\Gamma_0 f, \Gamma_1 g)_{\mathcal{K}},\quad f,g \in \dom\mathcal A^*.
    \end{equation*}
\item
  For any $Y_0, Y_1 \in{\mathcal{K}}$ there exists $f \in
\dom\mathcal A^*$, such that $\Gamma_0 f = Y_0$, $\Gamma_1 f = Y_1$.
In other words,
the mapping~$f \mapsto \Gamma_0 f \oplus \Gamma_1 f $, $f \in \dom\mathcal A^*$
to ${\mathcal{K}}\oplus {\mathcal{K}}$ is surjective.
\end{enumerate}
\end{definition}
\noindent
It can be shown~(see \cite{Kochubej}) that a boundary
triple for~$\mathcal A^*$
exists assuming only~$ n_{+}(\mathcal A)=n_{-}(\mathcal A) $.
Note also that a boundary triple
is not unique.
Given any bounded self-adjoint operator~$\Lambda = \Lambda^*$
on $\mathcal K$, the
collection~$\{\mathcal K, \Gamma_0 , \Gamma_1 + \Lambda \Gamma_0\}$
is a
boundary triple for~$\mathcal A^*$ as well,
provided that~$\Gamma_1 + \Lambda \Gamma_0$
is surjective.

\begin{definition}\label{def:WeylFunction}
Let $\mathscr T = \{\mathcal{K}, \Gamma_0, \Gamma_1 \}$
be a boundary triple
of~$\mathcal A^*$.
The Weyl function of $\mathcal A^*$ corresponding to~$\mathscr T$ and
denoted $M(z)$, $z \in\mathbb C\setminus \mathbb R$, is an analytic
operator-function with a positive imaginary part for~$z \in \mathbb C_+$
(i.e., an operator $R$\nobreakdash-function) with values in
the algebra of bounded operators on $\mathcal K$ such that
\[
  M(z)\Gamma_0 f_z = \Gamma_1 f_z, \quad f_z \in {\rm ker}(\mathcal A^* -zI),
  \quad z \notin \mathbb R.
\]
For $z\in\mathbb C \setminus\mathbb R$ we
have~$(M(z))^* = (M(\bar z))$ and
$\im (z) \cdot \im(M(z)) > 0$.
\end{definition}

\begin{definition}
An extension~$\mathscr A$ of a closed densely defined
symmetric operator~$\mathcal A$ is
called \emph{almost solvable (a.s.)} and denoted $\mathscr A = A_B$ if
there exist a boundary
triple~$\{\mathcal{K}, \Gamma_0, \Gamma_1 \}$ for~$\mathcal A^*$
and a bounded operator~$B : \mathcal{K} \to \mathcal K$ defined
everywhere in $\mathcal K$ such that
\begin{equation*}
	%\label{BoundaryCondition}
    f \in \dom A_B \iff  \Gamma_1 f =  B \Gamma_0 f
\end{equation*}
\end{definition}
\noindent
This definition implies the
inclusion~$\dom A_B \subset \dom \mathcal A^*$
and that~$A_B$ is a restriction of~$\mathcal A^*$ to
the linear set $\dom A_B := \{f \in \dom\mathcal A^* : \Gamma_1 f =
B \Gamma_0 f\}$.
In this context, the operator~$B$ plays the r\^ole of a
parameter for the
family of extensions~$\{A_B \mid B : \mathcal K \to \mathcal K\}$.
It can be shown (see~\cite{CherKisSilva1} for references)
that if the
deficiency indices $n_\pm(\mathcal A)$ are equal
and $A_B$ is an almost
solvable extension of~$\mathcal A$, then
the resolvent set of~$A_B$ is
not empty (i.e. $A_B$ is maximal),
both $A_B$ and $(A_B)^* = A_{B^*}$
are restrictions of $\mathcal A^*$ to their
domains, and
$A_B$ and $B$
are selfadjont (dissipative) simultaneously.
The spectrum of~$A_B$ coincides with the
set of points~$z_0 \in \mathbb C$ such that $(M(z_0) -B)^{-1}$
does not admit analytic continuation into it.

One of the cornerstones of our analysis is the celebrated Kre\u\i n formula, which allows to relate the resolvent of $A_B$  to the resolvent of a self-adjoint operator $A_\infty$ defined as the
%(self-adjoint)
restriction of the maximal operator $\mathcal A^*$ to the set
\begin{equation}\label{eq:DD}
{\rm dom}(A_\infty)=\bigl\{u\in {\rm dom\ }\mathcal A^*|\, \Gamma_0 u=0\bigr\}.
\end{equation}
(We follow \cite{Ryzhov} in using the notation $A_\infty,$ justified by the fact that in the language of triples this extension formally corresponds to $A_B$ with $B=\infty.$)

%The Kre\u\i n formula  obtained in \cite{DM}. It has been extended
%together with .

\begin{proposition}[Version of the Kre\u\i n formula of \cite{DM}]\label{prop:Krein}
Assume that $\{\mathcal{K},\Gamma_0,\Gamma_1\}$ is a boundary triple for the operator $\mathcal A^*$. Then for the resolvent  $(A_B-z)^{-1}$, where $B$ is a bounded operator in $\mathcal{K}$, one has, for all
%every $z$ such that
$z\in \rho(A_B)\cap \rho(A_\infty)$:
\begin{multline}
\label{eq:resolvent}
(A_B-z)^{-1}=(A_\infty-z)^{-1}+ \gamma(z)\bigl(B-M(z)\bigr)^{-1}\gamma^*(\bar z)
\\
=(A_\infty-z)^{-1}+ \gamma(z)\bigl(B-M(z)\bigr)^{-1}\Gamma_1 (A_\infty-z)^{-1},
\end{multline}
where $M(z)$ is the  M-function of $\mathcal A^*$ with respect to the boundary triple $\{\mathcal{K},\Gamma_0,\Gamma_1\}$ and $\gamma(z)$ is the solution operator
$$
\gamma(z)=\bigl(\Gamma_0|_{\text{\rm ker\,}(\mathcal A^*-z)}\bigr)^{-1}.
$$
\end{proposition}

Returning to the setup of Theorem \ref{th:direct}, we let $p$ denote the number of vertices of the fundamental graph  $\widetilde{\mathcal{G}}$. Labelling  these vertices by $\widetilde{v}_j$, $j=1,\ldots,p$,
we introduce $\widetilde{\Gamma}_0$ and $\widetilde{\Gamma}_1(t)$ to be linear operators from
\begin{equation}
\label{eq:ddomwtAmax}
\dom \wt{\Delta}_{\rm max}(A,t) = \{f\in W_2^2(\widetilde{\mathcal{G}}) \,\vert\, \, \text{ $f$ is continuous } \text{on the graph $\wtc{G}$ } \}
\end{equation}
to $\mathbb{C}^p$ defined as follows:
\begin{equation}
\label{eq:dGamma}
\widetilde{\Gamma}_0f=
\begin{pmatrix}
f(\widetilde{v}_1)\\
\ldots\\
f(\widetilde{v}_p)
\end{pmatrix};
\qquad
\widetilde{\Gamma}_1(t)f=
\begin{pmatrix}
\sum\limits_{\wt{e}\in\wtc{E}}\wt{\partial}^A_n f_{\wt{e}}(\wt{v}_1,t)
\\
\ldots
\\
\sum\limits_{\wt{e}\in\wtc{E}}\wt{\partial}^A_n f_{\wt{e}}(\wt{v}_p,t)
\end{pmatrix}.
\end{equation}

Then, performing integration by parts, we have for the maximal operator $\wt{\Delta}_{\rm max}(A,t)$, taking the role of $\mathcal{A}^*$ above and defined by the expression \eqref{eq:dwtA} on the domain \eqref{eq:ddomwtAmax},
\[
\langle \wt{\Delta}_{\rm max}(A,t) f ,g \rangle_{L_2(\widetilde{\mathcal{G}})}-\langle f ,\wt{\Delta}_{\rm max}(A,t) g \rangle_{L_2(\widetilde{\mathcal{G}})}
=
\langle \widetilde{\Gamma}_1(t)f, \widetilde{\Gamma}_0 g\rangle_{\mathbb{C}^p}-\langle \widetilde{\Gamma}_0 f,   \widetilde{\Gamma}_1(t)g\rangle_{\mathbb{C}^p},
\]
which implies

\begin{lemma}\label{lemma:triple}
The triple $\{\mathbb{C}^p, \widetilde{\Gamma}_0, \widetilde{\Gamma}_1(t)\}$ is a boundary triple in the sense of Definition  \ref{def:bonudaryTriple} for the operator $\wt{\Delta}_{\rm max}(A,t)$, defined by \eqref{eq:dwtA} on the domain \eqref{eq:ddomwtAmax}. The (minimal) symmetric operator $\mathcal A=\wt{\Delta}_{\rm min}(A,t)$ is defined by the same differential expression on the domain
\begin{equation}\label{eq:Amin}
\dom \wt{\Delta}_{\rm max}(A,t) \cap \ker \widetilde{\Gamma}_0 \cap \ker \widetilde{\Gamma}_1(t).
\end{equation}
\end{lemma}

\section{The graph $\mathcal{G}_1$}\label{sec:Model_graph}

In the present section we introduce our model operator and carry out its spectral analysis, aided by the general approach outlined above.
Consider the graph $\mathcal{G}_1$ represented in Fig.  \ref{fi:g1}, left. In order to visually emphasise that the graph is viewed as embedded into $\mathbb{R}^3$, we have elected to use dotted lines for the edges which would appear ``invisible'' if the graph was entirely located on an infinite cylinder. We will henceforth assume that all its edges are of unit length and that all the functions  $b_e$ are equal to identity, i.e.,
\[
l_{e}=1, \qquad b_{e}(x)\equiv1,  \qquad e\in \mathcal{E}.
\]
The corresponding fundamental graph  $\wtc{G}$ is represented in Fig.  \ref{fi:g1}, right. We will assume that the magnetic potential on the horizontal edges of $\mathcal{G}_1$ has opposite signs with its absolute value equal to a constant $A$, and that it is equal to zero identically on the rest of the edges.

\begin{remark}
Assuming that the graph  $\mathcal{G}_1$ is embedded into $\mathbb{R}^3$, the magnetic potential introduced above can be attained by a  \textbf{constant magnetic field}. Indeed, if one has the horizontal edges of the graph of Fig.  \ref{fi:g1} positioned along the lines  $y=\pm 1/2,z=0$, the rest of the edges lying in the planes parallel to the plane  $y=0$, then the 3D magnetic potential $\mathbf{A}=(-2A y,0,0)$ generates the prescribed magnetic potentials on the graph edges. The corresponding magnetic field is therefore
\[
\mathbf{B}=\mathbf{curl\ A} =(0,0,2A).
\]
It needs to be stressed that here we rely upon the general convergence results for Laplacians on thin structures converging to metric graphs, see \cite{Postbook} for details; cf. an alternative approach to the same problem developed in \cite{antiPost}. Indeed, the mentioned results are necessary in order to attribute precise meaning to the model of a Laplacian on a metric graph embedded into $\mathbb{R}^d$, both with and without a magnetic field.
\end{remark}

%куда вставить текст про то, что магнитное поле ок!
\begin{figure}[ht]
\vspace{0mm}
\centering

\unitlength 1.0mm
\linethickness{0.3pt}

\begin{picture}(100,30)(0,0)

\multiput(0,5)(20,0){4}{\circle*{1}}
\multiput(0,25)(20,0){4}{\circle*{1}}
\put(-5,5){\line(1,0){70.00}}
\put(-5,25){\line(1,0){70.00}}

\qbezier[100](0,5)(5,15)(0,25)
\qbezier[30](0,5)(-5,15)(0,25)

\qbezier[100](20,5)(25,15)(20,25)
\qbezier[30](20,5)(15,15)(20,25)

\qbezier[100](40,5)(45,15)(40,25)
\qbezier[30](40,5)(35,15)(40,25)

\qbezier[100](60,5)(65,15)(60,25)
\qbezier[30](60,5)(55,15)(60,25)

\qbezier[8](-5,5)(-7.5,5)(-10,5)
\qbezier[8](-5,25)(-7.5,25)(-10,25)
\qbezier[8](65,5)(67.5,5)(70,5)
\qbezier[8](65,25)(67.5,25)(70,25)

\put(100,5){\circle*{1}}
\put(100,25){\circle*{1}}
\qbezier[100](100,5)(105,15)(100,25)
\qbezier[30](100,5)(95,15)(100,25)
\put(100,2){\circle{6}}
\put(100,28){\circle{6}}

\put(20,-1){\vector(1,0){20}}
\put(40,28){\vector(-1,0){20}}

\put(28,0){$A$}
\put(28,29){$A$}

\put(70,15){$\mathcal{G}_1$}
\put(105,15){$\wtc{G}_1$}

\end{picture}

\vspace{0mm}
\caption{The graph $\mathcal{G}_1$  (left) and its fundamental graph  $\wtc{G}_1$ (right). Both solid and dotted lines represent graph edges.}
\label{fi:g1}
\end{figure}
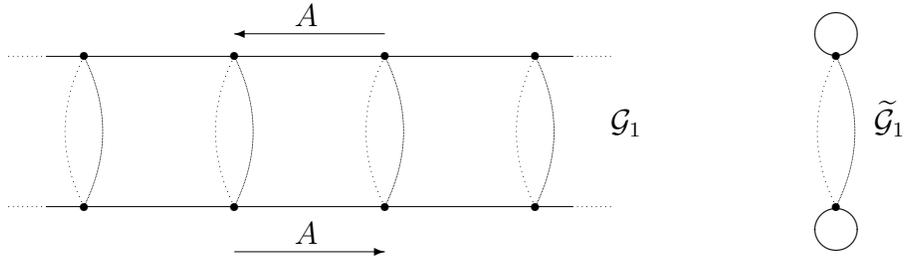

Denote by  $M(z,t,A)$ the $M$-matrix of the operator  $\widetilde{\Delta}(A,t)$ on the fundamental graph  $\wtc{G}_1$ for the magnetic potential $A$  introduced above and the quasimomentum  $t$, relative to the triple constructed in Lemma \ref{lemma:triple}. An explicit calculation in line with Definition \ref{def:WeylFunction}, \eqref{eq:dGamma} (see details in, e.g., \cite{Yorzh3}, cf. \cite{CEKN}) yields:
\begin{equation}
\label{eq:dwtM}
M(z,t,A)=\frac{2\sqrt{z}}{\sin \sqrt{z}}
\begin{pmatrix}
-2\cos \sqrt{z} +\cos(t+A) & 1 \\
1 & -2\cos \sqrt{z} +\cos(t-A)
\end{pmatrix}.
\end{equation}
\begin{lemma}
\label{le:specmin}
The non-residual spectrum of the minimal operator  $\wt{\Delta}_{\min}(A,t)$ on the graph $\wtc{G}_1$, defined by  \eqref{eq:dwtA} on the domain \eqref{eq:Amin}, consists of an infinite sequence of eigenvalues of multiplicity two,
\[
\sigma(\wt{\Delta}_{\min}(A,t))=\sigma_d(\wt{\Delta}_{\min}(A,t))=\{(\pi k)^2\}_{k=1}^{\infty},
\]
for all magnetic potentials $A$ of the form considered and all values of quasimomentum  $t\in[-\pi,\pi)$.
\end{lemma}
\begin{proof}
%почему чисто дискретный?

An eigenfunction  $u$ of the operator $\wt{\Delta}_{\min}(A,t)$, restricted to the edge $e$, solves the following spectral problem on $e$: \[
\begin{cases}
-\Big(\frac{d}{dx}+iB_e\Big)^2 u_e = \lambda u_e, \\
u(0)=u(1)=0,
\end{cases}
\]
where the quantity $B_e$ is determined by the magnetic potential and quasimomentum. The eigenfunctions and eigenvalues of this problem are
$
f_k=a_k e^{-iB_ex}\sin kx$ and $\lambda_k=(\pi k)^2, \ k\in\mathbb{N},
$
respectively,
where $a_k$ is an arbitrary constant. For the value $\lambda_k=(\pi k)^2$ to be an eigenvalue of the operator $\wt{\Delta}_{\min}(A,t)$ it suffices to pick constants  $a_k$ on different edges such that the Kirchhoff matching conditions at the graph vertices are satisfied, i.e., such that the eigenfunction corresponding to the value $\lambda_k=(\pi k)^2$ belongs to the kernel of the operator $\wt{\Gamma}_1$. In the case of the graph $\wtc{G}_1$ these conditions are represented as two linear equations to determine four independent variables. The linear independence of the named linear equations implies that the system is guaranteed to admit a solution; the dimension of the space of solutions is equal to two. Thus, $\lambda_k=(\pi k)^2$ is an eigenvalue of the operator  $\wt{\Delta}_{\min}(A,t)$ of multiplicity two, as claimed.
\end{proof}
\begin{remark}
Since the above lemma holds for all values of quasimomentum, it follows that the spectrum of the operator $\Delta_{\min}(A)$ on the graph $\mathcal{G}_1$, which is related to $\widetilde{\Delta}_{\min}(A,t)$ via \eqref{eq:fiber},  for any constant $A$ contains a discrete set of eigenvalues of infinite multiplicity $\lambda_k=(\pi k)^2$, $k\in\mathbb{N}$.
%\\
%{\bf (ii)} Among the eigenfunctions of the operator $\Delta_{\min}(A)$ corresponding to the eigenvalue $(\pi k)^2$, there exists an infinite set of compactly supported functions as well as an infinite set of non-compactly supported functions.
\end{remark}

\begin{theorem}
\label{th:spec}
The discrete spectrum of the operator $\Delta(A)$ on the graph $\mathcal{G}_1$ coincides with the spectrum of the minimal operator $\Delta_{\min}(A)$,
\[
\sigma_d(\Delta(A))=\{(\pi k)^2\}_{k=1}^{\infty}.
\]
The (absolutely) continuous spectrum of the operator $\Delta(A)$ on the graph $\mathcal{G}_1$ consists of an infinite set of spectral bands,  \[
\sigma_{ac}(\Delta(A))=\bigcup\limits_{k=0}^{\infty}\Big(\big[l^2_{k+},r^2_{k+}\big]\cup\big[l^2_{k-}  ,r^2_{k-} \big]\cup \big[(2\pi - r_{k-})^2,(2\pi - l_{k-})^2\big]\cup\big[(2\pi - r_{k+})^2 , (2\pi - l_{k+})^2\big]\Big),
\]
where the quantities  $l_{k\pm}$, $r_{k\pm}$ are determined as functions of the magnetic potential $A$  as follows:
\begin{align*}
l_{k+}&=
\begin{cases}
\arccos \frac{\sqrt{1+\sin^{-2} A}}{2}+2\pi k , \quad& \sin^2 A\geq |\cos A|,\\
\arccos \frac{1+|\cos A|}{2}+2\pi k, \quad& \text{otherwise},
\end{cases}
\\
r_{k+}&=\arccos \tfrac{1-|\cos A|}{2}+2\pi k,
\\
l_{k-}&=\arccos \tfrac{-1+|\cos A|}{2}+2\pi k,
\\
r_{k-}&=
\begin{cases}
\arccos\left( -\frac{\sqrt{1+\sin^{-2} A}}{2}\right) +2\pi k, \quad& \sin^2 A \geq |\cos A|,\\
\arccos\left( -\frac{1+|\cos A|}{2}\right) +2\pi k, \quad& \text{otherwise}.
\end{cases}
\end{align*}
The following subset of the spectrum of the operator $\Delta(A)$ is of multiplicity four:
\begin{equation}
\label{eq:specfour}
\bigcup\limits_{k=0}^{\infty}\Big(\big[l^2_{k+},l'^2_{k+}\big]\cup\big[r'^2_{k-}  ,r^2_{k-} \big]\cup \big[(2\pi - r_{k-})^2,(2\pi - r'_{k-})^2\big]\cup\big[(2\pi - l'_{k+})^2 , (2\pi - l_{k+})^2\big]\Big).
\end{equation}
Here
\begin{align*}
l'_{k+}&=
\arccos\left( \frac{1+|\cos A|}{2}\right)+2\pi k,
\\
r'_{k-}&=
\arccos\left( -\frac{1+|\cos A|}{2}\right) +2\pi k.
\end{align*}
The remaining part of the continuous spectrum has multiplicity equal to two.
\end{theorem}
\begin{remark}
Under the condition $|\cos A|>\sin^2 A$ the set \eqref{eq:specfour} is shown to be empty and therefore all the continuous spectrum of the operator is of multiplicity two. If on the other hand $\cos A =0$, then the set \eqref{eq:specfour} coincides with the continuous spectrum and the multiplicity of the continuous spectrum is equal to four everywhere.
\end{remark}

\begin{proof}
By \cite[Th. \RomanNumeralCaps{13}.85]{RS78}, instead of the spectrum of the operator $\Delta(A)$ one considers the spectra of the operator family  $\widetilde{\Delta}(A,t)$ for all $t\in[-\pi,\pi)$. Namely, $z\in\sigma(\Delta(A))$ iff one has
\begin{equation*}
\label{eq:rs}
\mu\{t\in[-\pi,\pi): (z-\delta,z+\delta)\cap\sigma(\Delta(A,t))\}\not=0,
\end{equation*}
where $\mu$ is the Lebesgue measure.
For the spectral analysis of the operator  $\widetilde{\Delta}(A,t)$ we consider its  $M$-matrix \eqref{eq:dwtM}.
Clearly, by \cite{BHS2020, Smudgen2012} the spectrum of $\widetilde{\Delta}(A,t)$ for a.a. $t$ can be obtained as the set
\[
\sigma(\wt{\Delta}_{\min}(A,t))\cup\{z\in\mathbb{R}:\, \det M(z,t,A)=0\},
\]
where the minimal operator is defined by \eqref{eq:dwtA}, \eqref{eq:Amin}, and its discrete spectrum is described in Lemma \ref{le:specmin}.
The eigenvalues of the matrix $M(z,t,A)$ are equal to
\[
\lambda_{\pm}(z,t,A)=\frac{2\sqrt{z}}{\sin \sqrt{z}}\Big(-2\cos \sqrt{z} + \cos A \cos t \pm \sqrt{1+\sin^2 A \sin^2 t}\Big).
\]
Consider the points $z$ such that  $\lambda_\pm(z,t,A)$ vanishes at some value of  $t\in[-\pi,\pi)$. We have
\begin{equation}
\label{eq:cos}
2\cos \sqrt{z} = \cos A \cos t \pm \sqrt{1+\sin^2 A \sin^2 t}.
\end{equation}
We remark that  $z$ is a continuous function of $t$ and each value of $z$ can only be attained a finite number of times. %In view of
% \eqref{eq:rs}, Lemma \ref{le:specmin}  and Remark 1 to the latter we have
%\begin{align}
%\sigma_d(\Delta(A))&=\{k^2\}_{k=1}^{\infty},
%\\
%\sigma_c(\Delta(A))&=\{z:\, \exists\,t_z\in[\pi,\pi], \, \det M(z,t_z,A)=0\},
%\end{align}
%причем кратность непрерывного спектра сопадает с числом таких различных $t_z$.

In order to complete the proof, we therefore only need to find out the ranges of the right hand sides in \eqref{eq:cos} considered as functions of $t$ for both signs in $\pm$ and all magnetic potentials $A$. Henceforth we abbreviate the mentioned right hand sides of \eqref{eq:cos} as $f_{\pm}(t,A)$.

Consider $f_+(t,A)$ first.
\\
1) The case of ${\sin^2 A>\cos A\geq0}$.
\\
In this case,  $f_{+}(t,A)$ has $4$ extrema on the set $[-\pi,\pi)$, attained at either
\[
\sin t=0 \quad\text{or}\quad \begin{cases}
\sin^2 t =\frac{\sin^4 A-\cos^2 A}{\sin^2 A}, \\ \cos t\geq 0.\end{cases}
\]
Two of these are global maxima with the function at them being equal to $\sqrt{1+\sin^{-2}A}$. The value of the function at the global minimum is equal to $1-\cos A$, whereas at the local minimum it admits the value $1+\cos A$.
Fig. \ref{fi:ranf} (a) represents the graph of $f_+$ for the magnetic potential $A=\frac{4\pi}{9}$.
\\
2) The case of ${\sin^2 A\leq\cos A}$.
\\
In this case, $f_+(t)$ has only two special points, namely  $t=0$ and $t=-\pi$.
The minimum and maximum are  $1-\cos A$ and $1+\cos A$, respectively.
Fig. \ref{fi:ranf} (b) is the graph of the function  $f_+$ for the magnetic potential $A=\frac{\pi}{9}$.
\begin{figure}[ht!]
\centering
\def\svgwidth{.9\textwidth}
\begingroup%
  \makeatletter%
  \providecommand\color[2][]{%
    \errmessage{(Inkscape) Color is used for the text in Inkscape, but the package 'color.sty' is not loaded}%
    \renewcommand\color[2][]{}%
  }%
  \providecommand\transparent[1]{%
    \errmessage{(Inkscape) Transparency is used (non-zero) for the text in Inkscape, but the package 'transparent.sty' is not loaded}%
    \renewcommand\transparent[1]{}%
  }%
  \providecommand\rotatebox[2]{#2}%
  \newcommand*\fsize{\dimexpr\f@size pt\relax}%
  \newcommand*\lineheight[1]{\fontsize{\fsize}{#1\fsize}\selectfont}%
  \ifx\svgwidth\undefined%
    \setlength{\unitlength}{465.67867878bp}%
    \ifx\svgscale\undefined%
      \relax%
    \else%
      \setlength{\unitlength}{\unitlength * \real{\svgscale}}%
    \fi%
  \else%
    \setlength{\unitlength}{\svgwidth}%
  \fi%
  \global\let\svgwidth\undefined%
  \global\let\svgscale\undefined%
  \makeatother%
  \begin{picture}(1,0.29183077)%
    \lineheight{1}%
    \setlength\tabcolsep{0pt}%
    \put(-0.00221451,0.24770262){\makebox(0,0)[lt]{\lineheight{1.25}\smash{\begin{tabular}[t]{l}(a)\end{tabular}}}}%
    \put(0.55825154,0.24770262){\makebox(0,0)[lt]{\lineheight{1.25}\smash{\begin{tabular}[t]{l}(b)\end{tabular}}}}%
    \put(0,0){\includegraphics[width=\unitlength,page=1]{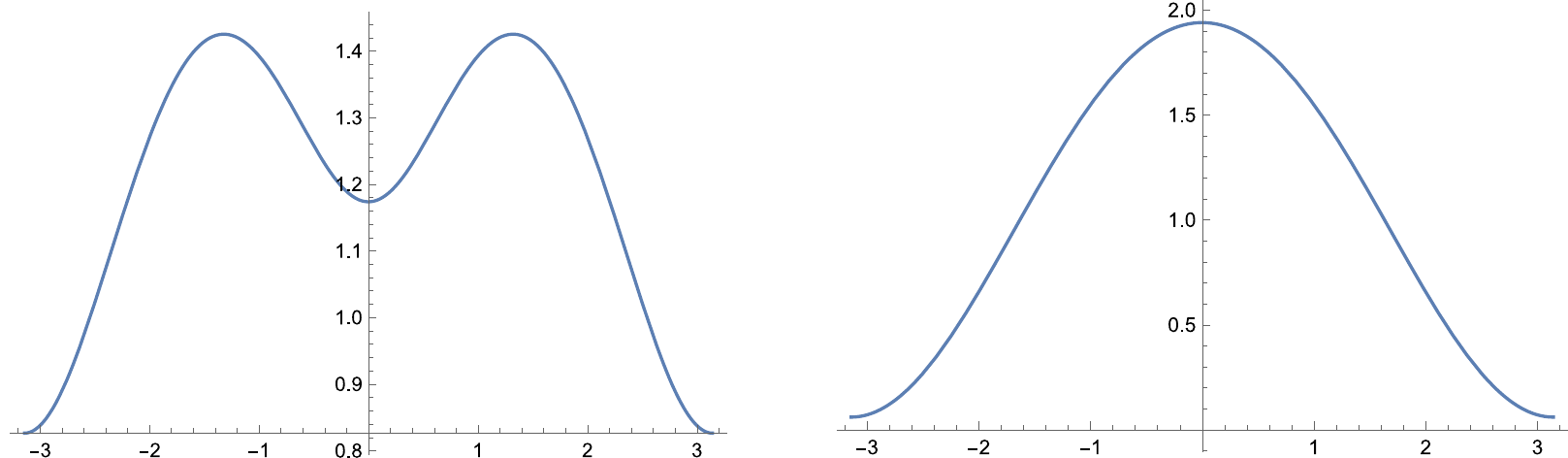}}%
  \end{picture}%
\endgroup%

\caption{The graph of the function  $f_+(t)$ for (a) $A=\frac{4\pi}{9}$; (b) $A=\frac{\pi}{9}$.}
\label{fi:ranf}
\end{figure}
\\
3) Other values of $A$.
\\
We have considered all $A\in[-\pi/2,\pi/2)$; moreover, one has
\[
M(z,t,A)=M(z,t\pm\pi,A-\pi).
\]
Hence for an arbitrary $A\in[\pi/2,3\pi/2)$ the set of ``zeros'' of the $M$-matrix coincides with the same for $A-\pi$, with the only difference that they are attained at the values  of quasimomentum, differing by  $\pi$.

\medskip

In order to consider the function $f_{-}(t)$ suffices to observe that
\[
f_-(t)=-f_+(t\pm\pi),
\]
from where the information on the range of $f_-$ is derived.
\end{proof}

\section{Rescaling}\label{sec:Rescaling}
In this section, we momentarily divert our attention from the example graph $\mathcal{G}_1$ to the general case. In the following sections, we will use the unitary transform defined below to perform the spectral analysis of the rescaled version of the operator of Section \ref{sec:Model_graph}.

Denote by $\mathcal{G}^{\varepsilon}=(\mathcal{V}^{\varepsilon},\mathcal{E}^{\varepsilon})$ the graph constructed based on  $\mathcal{G}=(\mathcal{V},\mathcal{E})$ so that
\begin{align*}
\mathcal{V}^{\varepsilon}&=\{\varepsilon v\,\vert\,v\in\mathcal{V}\},
\\
\mathcal{E}^{\varepsilon}&=\{(\varepsilon v,\varepsilon u)\,\vert\, (v,u)\in\mathcal{E}\},
\end{align*}
and the length of the segment corresponding to the edge  $(\varepsilon v, \varepsilon u)$ is equal to
\[
l_{(\varepsilon v, \varepsilon u)}=\varepsilon l_{(v,u)}.
\]
Note that if one treats the graph as a collection of points in $\mathbb{R}^{d_0}$, connected by smooth curves, the operation introduced above is nothing but rescaling of the graph by a factor of $\varepsilon$.

Let $\Delta_{\varepsilon}(A)$ be the magnetic graph Laplacian  \eqref{eq:dDelta}, defined on the graph  $\mathcal{G}^{\varepsilon}$.  The operator obtained by passing to its fiber decomposition will be denoted by $\widetilde{\Delta}_{\varepsilon}(A,t)$. Henceforth, it will be convenient to us to further unitary rescale this latter operator to the cell $\widetilde{\mathcal G}$. To this end, for $e$ being an edge of the graph $\widetilde{\mathcal{G}}$ and $\varepsilon e$ being the corresponding edge of the graph $\widetilde{\mathcal{G}}^\varepsilon$,
let $f\in L^2(\widetilde{\mathcal{G}})$. Denote by  $f_\varepsilon$ the following function on the graph  $\widetilde{\mathcal{G}}^\varepsilon$:
\[
(f_\varepsilon)_{\varepsilon e}(x) = f_e\Big(\frac{x}{\varepsilon}\Big).
\]
Introduce the operator $T_\varepsilon:L_2(\widetilde{\mathcal{G}}) \mapsto L_2(\widetilde{\mathcal{G}}^\varepsilon)$ by
\[
T_\varepsilon f = \frac{1}{\sqrt{\varepsilon}} f_\varepsilon.
\]
Since
\[
\Vert f_\varepsilon \Vert_{L_2(\widetilde{\mathcal{G}}^\varepsilon)}=\sqrt{\varepsilon}\Vert f \Vert_{L_2(\widetilde{\mathcal{G}})},
\]
the operator $T_\varepsilon$ is an isometry. Since it is also clearly surjective, $T_\varepsilon$ is unitary.

Slightly abusing the notation, in what follows we will keep using the same symbol $\widetilde{\Delta}_{\varepsilon}(A,t)$ for the operator $T_\e^* \widetilde{\Delta}_{\varepsilon}(A,t)  T_\e$.

Finally, the $M-$matrix of the latter operator at the value of the spectral parameter $z$ is $M_\varepsilon(z,A,t)$, in line with the notation of Section \ref{sec:Model_graph}.
\begin{lemma}\label{lemma:rescale}
The following relations hold.
\begin{align}
\label{eq:epsilon1}
\sigma(\Delta_\varepsilon (A/\varepsilon))&=\frac{1}{\varepsilon^2}\sigma(\Delta(A)),
\\
\label{eq:epsilon2}
\sigma(\widetilde{\Delta}_\varepsilon(A/\varepsilon,t/\varepsilon))&=\frac{1}{\varepsilon^2}\sigma\big(\widetilde{\Delta}(A,t)\big),
\\
\label{eq:epsilon3}
M_\varepsilon (z/\varepsilon^2,A/\varepsilon,t/\varepsilon)&=\frac{1}{\varepsilon^2}M(z,A,t).
\end{align}
\end{lemma}
\begin{proof}
The first two relations are trivial. As for the equality \eqref{eq:epsilon3}, it can be obtained along similar lines to the corresponding calculation in \cite{CEKN}, see also \cite{CEK} for a general PDE case.
%ниже не совсем чушь, но она не нужна для доказательства
\begin{comment}
Для этого обозначим через $\wtc{G}_{\varepsilon}$ -- фундаментальный граф для графа $\mathcal{G}_{\varepsilon}$. Пусть $U_\varepsilon$ -- это оператор \eqref{eq:dU} действующий из $L_2(\mathcal{G}_{\varepsilon})$ в  пространство $\int_{\widetilde{\mathcal{C}}/\varepsilon}\oplus L_2(\wtc{G}_{\varepsilon})\,dt$ (прямой интеграл пространств $L_2(\wtc{G}_{\varepsilon})$).
Пусть $\wt{f}\in \int_{\widetilde{\mathcal{C}}}\oplus L_2(\wtc{G}_{\varepsilon})\,dt$. Обозначим через $\wt{f}_\varepsilon$ следующее семейство функций на графе $\mathcal{G}_\varepsilon$:
\[
(\wt{f}_\varepsilon)_{\varepsilon e}(x,t) = \wt{f}_e\Big(\frac{x}{\varepsilon},\varepsilon t\Big).
\]

Наконец, обозначим через $\wt{T}_\varepsilon:\int_{\widetilde{\mathcal{C}}}\oplus L_2(\wtc{G})\,dt \to \int_{\widetilde{\mathcal{C}}/\varepsilon}\oplus L_2(\wtc{G}_{\varepsilon})\,dt$ оператор, действующий по правилу
\[
\wt{T}_\varepsilon \wt{f} = \varepsilon^{\frac{d-1}{2}} \wt{f}_\varepsilon.
\]
Такой оператор является унитарным. Явное вычисление показывает, что выполнено
\[
U_{\varepsilon} T_{\varepsilon}=\wt{T}_{\varepsilon} U.
\]
Отсюда, из теоремы \ref{th:direct} и равенства \eqref{eq:epsilon1} следует соотношение \eqref{eq:epsilon2}.
\end{comment}
%Сюда я планирую дописать что-то про $M$-матрицу, когда про нее будет написано в предыдущих параграфах.
\end{proof}

The result just proved now allows us to consider $\e$-periodic tubular structures as required, and in particular, their limiting behaviour as $\e\to 0$. In so doing, by a slight abuse of notation we shall keep referring to $A$ as the value of the magnetic potential, although in reality the magnetic potential is henceforth being scaled by $\e^{-1}$ in line with the statement of Lemma~\ref{lemma:rescale}.

\section{Degeneracy of band edges}\label{sec:Degeneracy}

Here we return to the model operator of Section \ref{sec:Model_graph}. We will now show that under a specific choice of the magnetic potential the band function of the operator  $\Delta_\varepsilon(A)$ (and in particular, in the limit of $\varepsilon\to0$, which will be used in Section \ref{sec:Homogenisation} below) admits an asymptotic behaviour as $(t-t_0)^4$ in a vicinity of the (lower) edge of its spectrum or indeed in a vicinity of an edge of any of its spectral bands. Here $t_0$ denotes the value of the quasimomentum corresponding to the spectral parameter at the edge of a spectral band. We will henceforth refer to this behaviour as \emph{the degeneracy of an edge of a spectral band}. Generically of course for second-order ordinary differential operators the asymptotics of the form $(t-t_0)^2$ takes place, hence the choice of terminology.

Denote by $A_\pm$ magnetic potentials such that the following equalities hold:
\begin{equation}\label{eq:degeneracy_condition}
\cos A_\pm = \pm \sin^2 A_\pm.
\end{equation}
We will henceforth refer to values of the magnetic potential $A$ satisfying this equation as \emph{critical} in view of the following theorem.
\begin{theorem}
\label{th:degeneration}
The lower edge of the spectrum of the operator $\Delta_\varepsilon(A)$ is degenerate iff $A=A_{\pm}$. Under this condition, the left edges of all odd and the right edges of all even bands are degenerate as well.
\end{theorem}
\begin{proof}
Introduce the notation $\tau =\varepsilon t$. Clearly, the rescaled quasimomentum $\tau$ takes values in the interval  $[-\pi,\pi)$. Consider the $M$-matrix of the operator $\wt{\Delta}_{\varepsilon}(A/\varepsilon,\tau/\varepsilon)$ on the graph $\wtc{G}_{\varepsilon}$, which by Lemma \ref{lemma:rescale} admits the form
\begin{equation*}
M_\varepsilon (z,A/\varepsilon,\tau/\e)=\frac{2\sqrt{z}}{\varepsilon\sin \varepsilon\sqrt{z}}
\begin{pmatrix}
-2\cos \varepsilon\sqrt{z} +\cos(\tau+A) & 1 \\
1 & -2\cos \varepsilon\sqrt{z} +\cos(\tau-A)
\end{pmatrix}.
\end{equation*}
The eigenvalues of the matrix $M(z,\tau)$ are computed as
\begin{equation}\label{eq:evs}
\lambda_{\pm}(z,A/\varepsilon,\tau/\e)=\frac{2\sqrt{z}}{\varepsilon\sin \varepsilon\sqrt{z}}\Big(-2\cos \varepsilon\sqrt{z} + \cos A \cos \tau \pm \sqrt{1+\sin^2 A \sin^2 \tau}\Big).
\end{equation}
The dispersion relation therefore admits the form
\begin{equation}
\label{eq:spedge}
2\cos \varepsilon \sqrt{z}  = \cos A \cos \tau \pm \sqrt{1+\sin^2 A \sin^2 \tau}
\end{equation}
for $\tau\in[-\pi,\pi)$. Then an edge of a spectral band corresponds to  $z$ such that the expression on the right hand side admits an extremum for one of the signs $\pm$. It follows that there could exist four groups of band edges. As in the proof of Theorem \ref{th:spec} we introduce the abbreviation
\[
f_\pm(t)=\cos A \cos \tau \pm \sqrt{1+\sin^2 A \sin^2 \tau}.
\]

Consider the function $\lambda_+$. The global extremum of the function $f_+$ can be attained at the following values of  $\tau$ only:
\begin{align*}
\text{a maximum at }&
\begin{cases}
\sin \tau=0,
\\ \cos A \cos \tau\geq 0
\end{cases}
\quad\text{or}\quad \begin{cases}
\sin^2 \tau =\frac{\sin^4 A-\cos^2 A}{\sin^2 A}, \\ \cos A \cos \tau\geq 0;\end{cases}
\\
\text{a minimum at }&
\begin{cases}
\sin \tau=0,
\\ \cos A \cos \tau\leq 0.
\end{cases}
\end{align*}
We next check whether any of these points yield degeneracy. Degeneracy of a band edge, denoted as $z_\e$, at $\tau=s$ is equivalent to the fact that the expansion of the function $\lambda_{\pm}(z_\varepsilon,A/\varepsilon,\tau/\varepsilon)$ in powers of $\tau-s$ has no terms of the first, second and third orders. We remark that the first order term is absent due to the dispersion relation. Abbreviate  $k_0:=\varepsilon \sqrt{z_\varepsilon}$ and denote the coefficients of the expansion of  $\lambda_{+}(z_\varepsilon,A/\varepsilon,\tau/\varepsilon)$ in powers of $\tau-s$ in a vicinity of an arbitrary  $s$ by $\mu_j(s,A)$. We have
\begin{align*}
\mu_0(s,A)&=0,
\\
\mu_1(s,A)&=\frac{2k_0 \sin s}{\sin k_0}\Big(\frac{\sin^2 A \cos s }{\sqrt{1+ \sin^2 A \sin^2 s}}-\cos A\Big),
\\
\mu_2(s,A)&=\frac{k_0}{\sin k_0}\Big(\frac{\sin^2 A \cos^2 s -\sin^2 A \sin^2 s-\sin^4 A \sin^4 s}{(1+ \sin^2 A \sin^2 s) ^{3/2}}-\cos A\cos s\Big),
\\
\mu_3(s,A)&=\frac{k_0}{3\sin k_0}\Big(-\frac{\sin^2 A \cos s \sin s (4+3\sin^2 A \cos^2 s + 5\sin^2 A \sin^2 s +\sin^4 A \sin^4 s )}{(1+ \sin^2 A \sin^2 s) ^{5/2}}\\
&+\cos A \sin s\Big).
\end{align*}
%Обозначим через $\text{sign}\,x$ функцию дающую $1$ при положительных $x$ и $-1$ при отрицательных.
In the case of $\sin s=0$ we obtain
\begin{align*}
\mu_1(s,A)&=0,
\\
\mu_2(s,A)&=\frac{k_0}{\sin k_0}(\sin^2 A - \cos A \sgn\cos t),
\\
\mu_3(s,A)&=0,
\end{align*}
where $\sgn x$ is the standard notation for the sign of $x$. Hence, for $s=0$ the band edges corresponding to $\lambda_+$ degenerate when $A=A_\pm$, moreover, this takes place at the value of $\tau$ corresponding to a \emph{maximum} of $f_+(t)$.

In the case of $\sin^2 s =\frac{\sin^4 A-\cos^2 A}{\sin^2 A}$ we have
\begin{align*}
\mu_1(s,A)&=0,
\\
\mu_2(s,A)&=\frac{k_0}{\sin k_0}\frac{1-2\sin^4 A - \sin^6 A}{\vert \sin^3 A \vert (1+\sin^2 A)^{3/2}}.
\end{align*}
In this case, the magnetic potentials at which the equality $\mu_2(s,A)=0$ holds are $A_\pm$ only. But then  $\sin^4 A - \cos^2 A=0$ and $\sin s=0$, whence  $\mu_3(s,A)=0$ at the same time. The degeneracy occurs at the values of  $\tau$, corresponding to a maximum of  $f_+(t)$, as above.

Ultimately, at the points $A=A_\pm$ left edges of $(1+4k)$th and right edges of $(4+4k)$th spectral bands, $k=0,1,\dots$, turn out to be degenerate; these edges correspond to maxima of  \eqref{eq:spedge} under the sign choice $+$. Right edges of  $(1+4k)$th and left edges of $(4+4k)$th spectral bands, corresponding to minima of the mentioned expression, are not degenerate for all values of the magnetic potential.

The result pertaining to $f_-(t)$ immediately follows from the relation
\[
f_-(t)=-f_+(t\pm\pi).
\]
Namely, one faces degeneracy only if $A=A_\pm$, but here the edges pertaining to minima of  $f_-(t)$ degenerate, that is, right edges of  $(2+4k)$th and left edges of  $(3+4k)$th,  $k=0,1,\dots$, spectral bands.
\end{proof}
\begin{remark}
The fourth order term of the asymptotic expansion $\lambda_{\pm}(z_\varepsilon,A_\pm/\varepsilon,\tau/\varepsilon)$ in the vicinity of band edges is non-zero. Hence the band function there admits either the form $(t-t_0)^2$ or the form $(t-t_0)^4$ for all values of the magnetic potential $A$. Degeneracy of higher orders is therefore impossible.
\end{remark}

The results of this section can be further compared to those of Section \ref{sec:Rescaling}. In particular, in relation to the lower edge of the spectrum one has the following picture. As the magnetic potential $A$ is increasing from zero, the lower edge of the spectrum is moving to the right; the spectrum, which is band-gap, is of local multiplicity two everywhere. This is the regime where, as $\e\to 0$, the Birman-Suslina homogenisation technique \cite{BSu03} is applicable. At the first critical value of $A$, see \eqref{eq:degeneracy_condition} (see also \eqref{eq:dA} below), the spectrum is still of multiplicity two, but the lower edge of it becomes degenerate in the sense introduced above. As $A$ increases further, a segment of the first band in the vicinity of the lower edge changes local multiplicity from 2 to 4, the mechanism for that is the change of the shape of the function $f_+$ from Fig. \ref{fi:ranf} (b) to Fig. \ref{fi:ranf} (a). The point $\tau/\e$ of quasimomentum, $\tau\in[-\pi,\pi)$, where the corresponding value of the spectral parameter meets the lower edge of the spectrum, is $\tau=0$ at the values of $A$ below the first critical one, splitting into two, $\tau=\tau_\pm$, $\tau_+=-\tau_-$, for higher values of $A$. As $A$ increases further, all of the absolutely continuous spectrum eventually assumes multiplicity 4, which corresponds to a situation where the global minima of Fig. \ref{fi:ranf} (a) become equal to the local minimum at zero. Increasing $A$ further still, the sub-bands of multiplicity 4 start shrinking, so that $\tau_+=\tau_-=0$ at precisely the second critical value of $A$. For higher values of $A$, one faces the classical regime yet again, and the situation starts repeating itself periodically thereafter. It turns out therefore that the critical values of the magnetic potential are exactly those at which the sub-bands of higher local multiplicity either appear or disappear.

The above description clearly carries over virtually unchanged to the behaviour of all left edges of odd-numbered spectral bands and all right edges of even-numbered ones.

As the above argument now clearly requires a detailed analysis of the model near the critical values of the magnetic potential to be carried out, we proceed with it in the following sections.

\section{Analysis in vicinity of a critical value of magnetic potential}\label{sec:Perturbation}

In the present section we will set the value of the magnetic potential to be equal to $A'$ with $A'=A+\delta$, where $A$ is a (critical) magnetic potential such that the equality \eqref{eq:dA} below holds and $\delta$ is an arbitrarily small number.

From now on we will also assume that $\varepsilon$ is arbitrarily small. Our immediate goal will be to derive asymptotics of the inverse to the $M$-matrix, which will then be utilized in constructing the norm-resolvent asymptotics of the operator family $\Delta_{\varepsilon}(A'/\varepsilon,\tau/\varepsilon)$ uniform with respect to the quasimomentum. The asymptotics of the inverse to the $M$-matrix sought will be shown to admit different forms depending on the sign of $\delta$. For $\delta\leq 0$, our ultimate goal is \eqref{eq:label} below, whereas in the case opposite we aim at proving \eqref{eq:labelpos}.

For simplicity, w.l.o.g. we consider the lower edge of the spectrum only, in the case where as $\delta$ increases, the local multiplicity of the named edge goes from 2 to 4. In particular, one can pick the following value of $A$:
\begin{equation}
\label{eq:dA}
A=\arccos \sqrt{\frac{\sqrt{5}-1}{\sqrt{5}+1}}.
\end{equation}
Since for such $A$ one clearly has
\begin{equation}
\label{eq:dAA}
\cos A = \sin^2 A,
\end{equation}
the left edge $z_\e$ of the first spectral band of the operator $\Delta_{\varepsilon}(A/\varepsilon)$ is attained at the quasimomentum $\tau=0$ and is degenerate. Different critical values of $A$ and internal edges of spectral bands can be considered likewise.

 Let
$z'_\varepsilon$ be the left edge of the spectrum of the operator $\Delta_{\varepsilon}(A'/\varepsilon,\tau/\varepsilon)$ on the graph $\wtc{G}_1^\varepsilon$ with magnetic potential $A'$. Let as above $k_0'=\varepsilon \sqrt{ z'_\varepsilon}$.

Further, since in proving norm-resolvent convergence results one obviously needs to distance oneself from the spectrum of both the original operator family and the effective, or homogenised, operator, we will assume throughout that
$$
z\in K_\sigma :=\{\lambda\in K\subset \mathbb{C} : {\rm dist } (\lambda, \mathbb{R})\geq \sigma >0\},
$$
where $K$ is some compact in the complex plane $\mathbb{C}$.

Consider the $M$-matrix $M_{\varepsilon} (z'_\e+z,A'/\varepsilon,\tau/\varepsilon)$  (see \eqref{eq:dwtM}) pertaining to the graph $\wtc{G}_1^\varepsilon$.

Assume first that $\delta\leq 0$. As the magnetic potential decreases from its critical value, the lower edge of the spectrum is attained at the same value of quasimomentum  $\tau=0$ (see the proof of Theorem \ref{th:spec}). For $k_0'$ we have
\[
2\cos k_0'=\cos A +1- \delta \sin A +O(\delta^2),
\]
which implies, in line with Lemma \ref{lemma:rescale}, that $z'_\varepsilon = O(\varepsilon^{-2})$ as $\varepsilon\to 0$.

Expanding the eigenvalues $\lambda_{\pm}(z'_\varepsilon+z,A'/\varepsilon,\tau/\varepsilon)$ of the $M$-matrix, see \eqref{eq:evs},  into the power series in a vicinity of the point $z=0$, $\tau=0$ (i.e., in a vicinity of the lower edge of the spectrum of the operator $\Delta_{\varepsilon}(A')$) and taking into account the identity \eqref{eq:dAA} we obtain
\begin{align*}
\lambda_+(z'_\varepsilon+z,A'/\e,\tau/\varepsilon)&=\frac{1}{\varepsilon^2}\frac{2k_0'}{\sin k_0'}
\Big(\frac{\sin 2A + \sin A}{2}\delta \tau ^2 - \frac{\tau^4}{8}+O(\tau^6)+O(\delta^2\tau^2)\Big)
\\
&+2z+O(\tau^4)+O(\delta \tau^2)+O(\varepsilon^2);
\\
\lambda_-(z'_\varepsilon+z,A'/\e,\tau/\varepsilon)&=\lambda_+(z'_\varepsilon+z,A'/\e,\tau/\e)-\frac{4k_0'}{\varepsilon^2\sin k_0'}+O(1).
\end{align*}

It is easily seen that the eigenvectors of the matrix $M_\varepsilon(z'_\varepsilon+z,A'/\varepsilon,\tau/\varepsilon)$ are orthogonal and after normalization admit the form which handily implies independence on $z,z'_\e$, namely,
\begin{equation}\label{eq:nuvec}
%\begin{equation}
\nu_{\pm}=\frac{1}{\sqrt{2}}
\begin{pmatrix}
\sqrt{1\pm\dfrac{\sin A \sin \tau}{\sqrt{1+\sin^2 A \sin^2 \tau}}} \\[15pt]
\mp \sqrt{1\mp\dfrac{\sin A \sin \tau}{\sqrt{1+\sin^2 A \sin^2 \tau}}}
\end{pmatrix}.
%\end{equation}
\end{equation}
The asymptotic expansion of $M^{-1}(z'_\varepsilon+z,A'/\varepsilon,\tau/\varepsilon)$ in powers of $\varepsilon$ as $\varepsilon\to0$
is then computed in the basis of $\nu_{\pm}$ as
\begin{align}
\label{eq:label}
M_{\varepsilon}^{-1}(z'_\varepsilon+z,A'/\varepsilon,\tau/\varepsilon)&=
\begin{pmatrix}
\dfrac{1}{2z+\frac{k_0'(\sin 2A + \sin A)}{\sin k_0'} \frac{\delta \tau^2}{\varepsilon^2}-\frac{k_0'}{4 \sin k_0' } \frac{\tau^4}{\varepsilon^2}}  & 0 \\
0 & 0
\end{pmatrix}\\
&+O(\min\{\varepsilon,\varepsilon^2/|\delta|\})+O(|\delta|).\nonumber
\end{align}
In proving the above asymptotic expansion we consider the difference
$$
\lambda_+(z'_\varepsilon+z,A'/\e,\tau/\varepsilon)^{-1}- \dfrac{1}{2z+\frac{k_0'(\sin 2A + \sin A)}{\sin k_0'} \frac{\delta \tau^2}{\varepsilon^2}-\frac{k_0'}{4 \sin k_0' } \frac{\tau^4}{\varepsilon^2}}
$$
and split the domain of $\tau$ into segments $|\tau|\leq \e^\gamma$, $\e^\gamma\leq|\tau|\leq |\delta|^{1/2}$, $|\delta|^{1/2}\leq|\tau|\leq c_0$ for some fixed small enough $c_0$ and the parameter $\gamma>0$ to be picked later, and the remaining part of the dual cell. In each of the named regions the estimate sought is obtained by the standard asymptotic analysis, see e.g. \cite{CEKN} for details of a similar computation. In particular, for $|\tau|\leq \e^\gamma$ one uses the assumption that $z\in K_\sigma$.  The parameter $\gamma$ is then chosen to optimise the estimates thus obtained.

The estimate \eqref{eq:label} provides us with the desired asymptotic formula for $\delta\leq 0$.  Note that at the critical value of magnetic potential, i.e., at the point $\delta=0$, the leading order term of the above asymptotics becomes purely fourth order in $\tau$.

It is not unexpected given that the spectral structure of the operator considered (and in particular, the local multiplicity of the bands' edges) is different at values of $A'$ to the left and to the right of the critical one that the derivation of the asymptotics at $\delta>0$ diverges from the one given above. We now pass over to the consideration of this case.

As the magnetic potential $A'$ increases from its critical value $A$, the lower edge of the spectrum is attained at two values of quasimomentum $\tau=\tau_\pm$ such that
\begin{equation*}
\sin^2 \tau_\pm =\frac{\sin^4 A'-\cos^2 A'}{\sin^2 A'} \text{ and } \cos \tau_\pm\geq 0,
\end{equation*}
whence
\begin{equation}\label{eq:dtpm}
\tau_\pm=\pm\sqrt{2\delta}\sqrt{\sin 2A + \sin A}+O(\delta^{3/2}).
\end{equation}

Proceeding as in the case of $\delta<0$, one then obtains two estimates of the type \eqref{eq:label} independently for $\tau<0$ and $\tau>0$, where in the former $\tau^2$ in \eqref{eq:label} gets replaced by $(\tau-\tau_-)^2$, whereas in the latter it is substituted by $(\tau-\tau_+)^2$. This is clearly due to the fact that Fig. \ref{fi:ranf}, (a) implies that no single simple rational function can fully account for the asymptotic behaviour of $\lambda_+^{-1}$ in the case considered.

We remark that, as it will transpire in Section 9 below, this implies that the operator family $\Delta_\e(A'/\e)$ can generically converge to no scalar differential operator in the norm-resolvent topology when $A'>A$.

Nevertheless, for small enough $\delta$ the situation can be rectified. Indeed, let $\delta=O(\e^2)$.  Then from \eqref{eq:dtpm} we infer that $\tau_+-\tau_-=O(\e)$ which can be seen as the two maxima of Fig. \ref{fi:ranf} blending together. This allows the asymptotics for $\tau>0$ to be matched with that for $\tau<0$ to yield
\begin{align}
\label{eq:labelpos}
M_{\varepsilon}^{-1}(z'_\varepsilon+z,A'/\varepsilon,\tau/\varepsilon)&=
\begin{pmatrix}
\dfrac{1}{2z-\frac{k_0'}{4 \sin k_0' } \frac{\tau^4}{\varepsilon^2}}  & 0 \\
0 & 0
\end{pmatrix}+O(\varepsilon^{1/2})
\end{align}
for $\delta>0$ such that $\delta=O(\e^2)$.

The seeming discrepancy between \eqref{eq:label} and \eqref{eq:labelpos} in that the latter does not contain the term of the order $\delta\tau^2$ is due to the fact that as it is easily seen the latter contributes an $O(\varepsilon^{1/2})$ error provided that $\delta=O(\e^2)$.

The estimate \eqref{eq:label} is therefore extended in a continuous way from $\delta<0$ into an $\e^2$-sized region of positive $\delta$, albeit at a cost of a significantly worse error bound.

%In the lower orders in $\delta$ the coefficients in front of powers of $(\tau-\tau')$ in the above expansions coincide with the corresponding coefficients in the expansion in powers of $\tau$ for  $\delta<0$. This can be clearly traced back to the fact that $\delta\to0$ yields  $\tau'\to0$. Therefore the asymptotic expansion of the $M$-matrix \eqref{eq:label} is modified to cover the case $\delta>0$ by replacing $\tau$ by $(\tau-\tau')$.

\section{Homogenisation of the $\e$-periodic tubular structure}\label{sec:Homogenisation}

In this section we apply the asymptotic expansion of the $M$-matrix obtained in Section \ref{sec:Perturbation} to obtain norm-resolvent convergence, as $\e\to 0$, of our model operators to their effective, or homogenised, counterparts. The mentioned convergence holds up to a unitary transform, providing us with a certain freedom of choosing the effective model in a physically motivated way. When the period of the medium becomes negligible, the model of Sections \ref{sec:Model_graph} and \ref{sec:Rescaling}, although still possessing its microstructure, looks as an infinite wire to a far-away macroscopic observer. It is for this reason that we will be picking an ODE on the real line to serve as a realisation of the effective operator. Since in dimension one the magnetic field is eliminated by a unitary gauge transform, it comes as no surprise that our homogenised model operator is non-magnetic.

The argument goes as follows. We first apply the Krein formula \eqref{eq:resolvent} of Proposition \ref{prop:Krein} to the operator family $\widetilde{\Delta}_\e(A'/\e,\tau/\e)$, where as in Section \ref{sec:Perturbation} we assume that $A'=A+\delta$, and $A$ is a magnetic potential such that the equality \eqref{eq:dA} holds, $\delta$ serving as an arbitrarily small offset, $\delta\leq 0$ or $0<\delta=O(\e^2)$. Further, the spectral parameter $z'_\e+z$, $z\in K_\sigma$ is fixed in a ``vicinity'' of the lower edge of the spectrum $z_\e'$. The requirement that $z\in K_\sigma$ is imposed here to ensure that the spectral parameter is separated from the spectrum, in line with both \cite{BirmanHF1,BirmanHF2} and \cite{BSu03}, in the latter mentioned paper the lower edge of the spectrum $z'_\e$ being at zero.

We express the solution operator $\gamma(z),$ see Proposition \ref{prop:Krein}, in terms of $\gamma(0)$ and the resolvent of the Dirichlet decoupling via  $\gamma(z)=(1-z(\widetilde{\Delta}_{\e,\infty}(A'/\e,\tau/\e))^{-1})^{-1}\gamma(0)$. The proof of this formula is standard and can be found in, e.g., \cite{Ryzhov}. Here the operator $\widetilde{\Delta}_{\e,\infty}(A'/\e,\tau/\e)$ is nothing but the Dirichlet decoupling of the operator $\widetilde{\Delta}_\e(A'/\e,\tau/\e)$, i.e., a self-adjoint operator defined by the same differential expression as the latter subject to Dirichlet boundary conditions at all vertices of the underlying fundamental graph. On the formal level, it is the rescaling of Section \ref{sec:Rescaling} applied to the self-adjoint extension $\widetilde{\Delta}_\infty (A',\tau)$ of the minimal operator described by Lemma \ref{lemma:triple} with the domain
$$
\dom \wt{\Delta}_{\rm max}(A,t) \cap \ker \widetilde{\Gamma}_0.
$$
The estimates
$$
\left\|\left(\widetilde{\Delta}_{\e,\infty}(A'/\e,\tau/\e) - (z'_\e+z)\right)^{-1}\right\|=O(\e^2),\qquad
\|\gamma(z'_\e+z)-\gamma(z'_\e)\|=O(\e^2)
$$
are obtained by the spectral theorem for a self-adjoint operator, see, e.g., \cite{CEKN} for details. They hold uniformly in $\tau$ and $z\in K_\sigma$, which follows easily from the fact that by \eqref{eq:cos} the spectra of $\widetilde{\Delta}(A',\tau)$ and $\widetilde{\Delta}_\infty (A',\tau)$ are separated uniformly in $\tau$ for sufficiently small values of $\delta$, say, $|\delta|<\delta_0$, and Lemma \ref{lemma:rescale}.

The Krein formula \eqref{eq:resolvent} then yields the estimate
$$
\left(\widetilde{\Delta}_{\e}(A'/\e,\tau/\e)- (z'_\e+z)\right)^{-1}=- \gamma(z'_\e) M_\e^{-1}(z'_\e+z,A'/\e,\tau/\e) \gamma(z'_\e)^*
+O(\e^2),
$$
which is uniform in the operator-norm topology for all $\tau\in[-\pi,\pi)$ and $z\in K_\sigma$.

Note that the first term on the right hand side of the last equality is a rank 2 operator, since such is the matrix $M_\e$. The asymptotic formulae for $M_\e^{-1}$ obtained in Section \ref{sec:Perturbation} further show that asymptotically this operator is rank 1. The leading order term of the last asymptotic equality is therefore essentially a scalar function, which we make explicit as follows.

Let $\mathcal{P}(\tau)$ be the orthogonal projection in the boundary space $\mathbb{C}^2$ onto the subspace spanned by the vector $\nu_+$, see \eqref{eq:nuvec}. Define a rank one operator $\breve\Pi(z'_\e,\tau)$ from $\mathbb{C}^2$ to $L_2(\widetilde{\mathcal{G}}_1)$ as follows:
$$
\breve\Pi(z'_\e,\tau):=\gamma(z'_\e) \mathcal{P}(\tau).
$$
By an explicit computation, one has
\begin{equation}\label{eq:as1}
\|\breve\Pi(z'_\e,\tau)\|^2 =: \kappa_{\tau,k'_0,A'}^2= \kappa_{0,k'_0,A'}^2 + O(\tau),
\end{equation}
where
%the notation
%$$
%\kappa_0^2:=\frac{1}{2 (1+\cos k'_0)} \left( 1+ \frac{\sin k'_0 \cos k'_0}{k'_0} \right)
%$$
%has been introduced and
$k'_0=\e \sqrt{z'_\e}$, as usual. The quantity $\kappa_{0,k'_0,A'}$ only depends on the value of the magnetic potential $A'$, since the lower edge of the spectrum $z'_\e$ is uniquely determined by the latter. Its explicitly computed value is omitted here due to the bulkiness of the expression. Note that unlike \cite{CEKN} the error in the above formula is estimated as $O(\tau)$ rather than $O(\tau^2)$, which is attributed to the presence of the magnetic potential.

Further, one has by yet another explicit calculation taking into account \eqref{eq:dA} and the fact that for the chosen value of $A$ one clearly has $k_0=\pi/5$ (recall that as before we have $k_0=\varepsilon \sqrt{z_\varepsilon}$, where $z_\e$ is the lower edge of the spectrum of the operator $\widetilde{\Delta}_\e(A/\e,\tau/\e)$):
\begin{equation}\label{eq:as2}
\kappa_{0,k'_0,A'}^2= \kappa_0^2 + O(\delta)
\end{equation}
with a rather extravagantly looking yet explicit independent constant
\begin{equation}\label{eq:rubbish}
  \kappa_0^2:=\frac{2 \left(\sqrt{5}+7\right) \pi -20 \sqrt{10-2 \sqrt{5}} }{\left(5-\sqrt{5}\right) \pi }\approx 1.26787.
\end{equation}

Using the asymptotics obtained in Section \ref{sec:Perturbation} we now have
\begin{equation}\label{eq:ref}
\left(\widetilde{\Delta}_{\e}(A'/\e,\tau/\e)- (z'_\e+z)\right)^{-1}=- \breve\Pi(z'_\e,\tau) \dfrac{1}{2z+\sigma(A,\tau,\delta,\e)} \breve\Pi(z'_\e,\tau)^* +r(\e,\delta),
\end{equation}
where
$$
r(\e,\delta)=\begin{cases}
O(\min\{\varepsilon,\varepsilon^2/|\delta|\})+O(|\delta|),& \delta\leq0,\\
O(\varepsilon^{1/2}),& \delta=O(\e^2)>0
\end{cases}
$$
and
$$
\sigma(A,\tau,\delta,\e)=\begin{cases}
\frac{k_0'(\sin 2A + \sin A)}{\sin k_0'} \frac{\delta \tau^2}{\varepsilon^2}-\frac{k_0'}{4 \sin k_0' } \frac{\tau^4}{\varepsilon^2},& \delta\leq0,\\
-\frac{k_0'}{4 \sin k_0'} \frac{\tau^4}{\varepsilon^2},&
\delta=O(\e^2)>0.
\end{cases}
$$

The operator on the right hand side is, as mentioned above, rank one. We introduce the vector $\psi(\tau)\in {\rm ran\ } \breve\Pi(z'_\e,\tau)$ such that $\|\psi(\tau)\|=1$ for all $\tau$. The equality \eqref{eq:ref} thus admits the form
$$
\left(\widetilde{\Delta}_{\e}(A'/\e,\tau/\e)- (z'_\e+z)\right)^{-1}=- \|\breve\Pi(z'_\e,\tau)\|^2 \dfrac{1}{2z+\sigma(A,\tau,\delta,\e)} \langle \cdot,\psi(\tau)\rangle \psi(\tau) +r(\e,\delta).
$$

We now use \eqref{eq:as1}, \eqref{eq:as2} to obtain by a straightforward estimate
$$
\left(\widetilde{\Delta}_{\e}(A'/\e,\tau/\e)- (z'_\e+z)\right)^{-1}=- \kappa_0^2 \dfrac{1}{2z+\sigma(A,\tau,\delta,\e)} \langle \cdot,\psi(\tau)\rangle \psi(\tau) +r_0(\e,\delta)
$$
with
$$
r_0(\e,\delta)=\begin{cases}
O(\min\{\varepsilon^{1/2},\varepsilon/|\delta|\})+O(|\delta|),& \delta\leq0,\\
O(\varepsilon^{1/2}),& \delta=O(\e^2)>0.
\end{cases}
$$

Ultimately, we pick a fiber-wise unitary gauge which will enable us to invoke a suitable inverse Gelfand transfer so that the leading order asymptotics derived thus far naturally maps to an ordinary differential operator on the real line.

Following \cite{CEKN}, to this end for each $\tau$ we select a unitary operator $U_\tau$ mapping the function $\psi(\tau)$ on the graph $\wtc{G}_1$ to the normalised unity function $\e^{-1/2}\pmb 1$ on the interval $[0,\e]$. This yields the direct integral over quasimomentum of the fiber operators
$$
\dfrac{1}{-2\kappa_0^{-2}z-\kappa_0^{-2}\sigma(A,\tau,\delta,\e)} \langle \cdot,\e^{-1/2}\pmb 1\rangle \e^{-1/2}\pmb 1 +\tilde r_0(\e,\delta),
$$
where the error terms $\tilde r_0(\e,\delta)$ admit the same uniform in $\tau$ estimate as $r_0(\e,\delta)$. Transforming it further by the (inverse) Gelfand transform associated with the infinite chain graph with links of length $\e$ (it is convenient here to think that the mentioned chain graph is naturally embedded into $\mathbb R$), by repeating the argument of \cite[Section 10]{CEKN} we obtain the following theorem.

\begin{theorem}\label{thm:MAIN}
Let $A$ be the critical value of the magnetic potential defined by \eqref{eq:dA}. Let $A'-A=\delta$ and $z'_\e$ be the lower edge of the spectrum of the operator ${\Delta}_\e(A'/\e)$; $k'_0=\e \sqrt{z'_\e}$.

Let $\mathcal{A}_-$ and $\mathcal{A}_+$ be self-adjoint operators in $L^2(\mathbb R)$ defined by the differential expressions
$$
\kappa_0^{-2}\left(\frac{k_0'(\sin 2A + \sin A)}{\sin k_0'} \frac{\delta d^2}{dx^2}+\frac{k_0'}{4 \sin k_0' } \frac{d^4}{dx^4}\right)
$$
and
$$
\kappa_0^{-2}\left( \frac{k_0'}{4 \sin k_0'} \frac{d^4}{dx^4} \right),
$$
respectively.

Then the following estimates hold uniformly in $z\in K_\sigma$ with a unitary operator $\Phi:L^2(\mathcal{G}^\e_1)\mapsto L^2(\mathbb{R})$:
  $$
  \|({\Delta}_\e(A'/\e)-(z'_\e+z))^{-1}-\Phi^*(\mathcal A_{-}-2\kappa_0^{-2}z)^{-1}\Phi\|= O((\min\{\varepsilon^{1/2},\varepsilon/|\delta|\})+O(|\delta|),\quad \delta\leqslant0,
  $$
and
  $$
  \|({\Delta}_\e(A'/\e)-(z'_\e+z))^{-1}-\Phi^*(\mathcal A_{+}-2\kappa_0^{-2}z)^{-1}\Phi\|= O(\varepsilon^{1/2}),\quad \delta=O(\e^2)>0.
  $$
Here $\kappa_0$ is defined by \eqref{eq:rubbish} and we also have $k'_0=k_0+O(\delta)$ with $k_0=\pi/5$.

\end{theorem}

\begin{remark}
1. We reiterate that this result continues to hold for all critical values of the magnetic potential $A$, see \eqref{eq:degeneracy_condition} and in vicinity of the left edges of all odd and of the right edges of all even bands of the spectrum of ${\Delta}_\e(A'/\e)$. The numerical value of $\kappa_0$ as well as the r\^ole of the sign of $\delta$ are of course subject to the required adaptation.

2. The unitary operator $\Phi$ appearing in the statement of the last theorem is rather involved, which is to be expected since its ``task'' is to reduce the complicated $\e$-periodic internal geometry of the tubular structure to the real line, cf. \cite{CherErKis}, where a similar unitary transform was utilised. An explicit form of this operator can be given as follows: $\Phi = U_\e^* (U_\tau\oplus V) U T_\e^*$. Here $T_\e$ is the unitary rescaling operator introduced in Section \ref{sec:Rescaling}, $U$ is the unitary Gelfand transform of Section \ref{sec:Fibre_representation} pertaining to the graph $\mathcal{G}_1$, $U_\tau$ is as above the unitary operator mapping the function $\psi(\tau)$ on the graph $\wtc{G}_1$ to the normalised unity function $\e^{-1/2}\pmb 1$ on the interval $[0,\e]$, $V$ is an arbitrary unitary operator mapping  $\{\psi(\tau)\}^\bot$ in $L^2(\widetilde{\mathcal{G}}_1)$ to $\{\e^{-1/2}\pmb 1\}^\bot$ in $L^2([0,\e])$, and, finally, $U_\e$ is the Gelfand transform of Section \ref{sec:Fibre_representation} applied to the infinite chain graph embedded into $\mathbb{R}$ with all edge lengths equal to $\e$. Despite the obvious complexity of this definition, it is notable that the essential ingredient of $\Phi$ is exactly the rank-one unitary gauge $U_\tau$, responsible for ``replanting'' a particular linear function from the fundamental graph of $\mathcal{G}_1$ to the fundamental graph of an $\e$-periodic chain graph.

3. We remark that at larger positive $\delta$ it is impossible to obtain a norm-resolvent asymptotics to a scalar ordinary differential operator. The underlying reason for that is that in this situation the dispersion relation of the homogenised medium admits a shape which can be only attained by considering a matrix operator. We conjecture that the proper effective model in this case can be constructed using the set of techniques introduced in \cite{survey}.

4. As mentioned in Section \ref{sec:Perturbation}, for $\delta\le 0$ such that $\delta=O(\e^2)$ the second-order part of the homogenised ODE can be dropped. In this case, $\mathcal{A}_-$ can be replaced by $\mathcal{A}_+$; the error bound $O(\sqrt{\e})$ stays sharp.

5. The estimate provided by Theorem \ref{thm:MAIN} only applies to the transition regime from the second-order operator to the fourth-order one, not covering the ``standard'' homogenisation regime for negative $\delta=O(1)$, where one obtains the classical second-order limit operator. Using the approach developed in the paper, it seems feasible that a unified estimate that would describe the continuous transition between pre-critical and critical cases could be obtained. The merit of rather tedious work required to do so is however questionable, since the classical homogenisation is a very well-studied area. On the other hand, the situation to the right of the critical value of the magnetic potential ($\delta>0$) appears much more interesting, giving rise to a matrix ordinary differential operator as the homogenisation limit and, on this basis, to the regimes with negative group velocity, usually attributable to metamaterial-type behaviour. This analysis falls beyond the scope of this paper and will appear elsewhere.

\end{remark}

\section*{Acknowledgements}
AVK's work was supported by the EPSRC grant EP/V013025/1 ``Quantitative tools for upscaling the micro-geometry of resonant media''. KR's work was supported by the grant of CONACyT CF-2019 No.\,304005.

AVK is grateful to the research team at IIMAS-UNAM, and in particular to Dr Luis O. Silva, for their
kind hospitality during his stay in Mexico from March to October 2022.

Both authors are grateful to their anonymous referees for providing valuable comments.

\emph{Data access statement.} No new data were generated or analysed during this study.

\end{document}